\documentclass{article}

\usepackage{arxiv}

\usepackage{natbib}
\usepackage{times}
\usepackage{float}
\usepackage{soul}
\usepackage{url}
\usepackage{hyperref}
\usepackage{caption}
\usepackage{graphicx}
\usepackage{amsthm}
\usepackage{booktabs}
\usepackage{amsmath,amssymb,amsfonts}
\usepackage{subcaption}
\usepackage{bbm}
\usepackage{physics}
\usepackage{bm}
\usepackage{framed}
\usepackage{multirow}
\usepackage{xcolor}
\usepackage{blindtext}
\usepackage{graphicx}
\usepackage{textcomp}
\usepackage[noend]{algpseudocode}
\usepackage{algorithm}
\theoremstyle{definition}
\newtheorem{thm}{Theorem}

\title{Robust Fraud Detection via Supervised Contrastive Learning}

\author{ \href{https://orcid.org/0000-0002-2595-2759}{\includegraphics[scale=0.06]{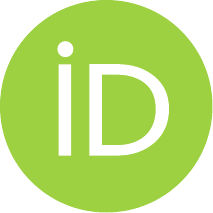}\hspace{1mm}Vinay M.S.}\\
	University of Arkansas\\
	Fayetteville, AR 72701, USA \\
	\texttt{vmadanbh@uark.edu} \\
	\And
	\href{https://orcid.org/0000-0001-6816-419X}{\includegraphics[scale=0.06]{orcid.pdf}\hspace{1mm}Shuhan Yuan} \\
	Utah State University\\
	Logan, UT 84322, USA \\
	\texttt{Shuhan.Yuan@usu.edu} \\
        \And
	\href{https://orcid.org/0000-0002-2823-3063}{\includegraphics[scale=0.06]{orcid.pdf}\hspace{1mm}Xintao Wu} \\
	University of Arkansas\\
	Fayetteville, AR 72701, USA \\
	\texttt{xintaowu@uark.edu} \\
}

\date{}

\renewcommand{\shorttitle}{Robust Fraud Detection via Supervised Contrastive Learning}

\hypersetup{
pdftitle={A template for the arxiv style},
pdfsubject={q-bio.NC, q-bio.QM},
pdfauthor={David S.~Hippocampus, Elias D.~Striatum},
pdfkeywords={First keyword, Second keyword, More},
}

\begin{document}
\maketitle

\begin{abstract}
Deep learning models have recently become popular for detecting malicious user activity sessions in computing platforms. In many real-world scenarios, only a few labeled malicious and a large amount of normal sessions are available. These few labeled malicious sessions usually do not cover the entire diversity of all possible malicious sessions. In many scenarios, possible malicious sessions can be highly diverse. As a consequence, learned session representations of deep learning models can become ineffective in achieving a good generalization performance for unseen malicious sessions. To tackle this open-set fraud detection challenge, we propose a robust supervised contrastive learning based framework called \textit{ConRo}, which specifically operates in the scenario where only a few malicious sessions having limited diversity is available. ConRo applies an effective data augmentation strategy to generate diverse potential malicious sessions. By employing these generated and available training set sessions, ConRo derives separable representations w.r.t open-set fraud detection task by leveraging supervised contrastive learning. We empirically evaluate our ConRo framework and other state-of-the-art baselines on benchmark datasets.  Our ConRo framework demonstrates noticeable performance improvement over state-of-the-art baselines.
\end{abstract}


\section{Introduction}

Computing platforms, such as social networking sites and cloud systems, experience large volumes of malicious or fraudulent activities due to the anonymity and openness of the Internet.
It is critical to identify such malicious activities in order to protect legitimate users. In practice, the activities of a user are usually modeled as an activity session. For example, in a computer system, an activity session is a sequence of user activities starting with log-in and ending with log-out. A popular approach for detecting malicious sessions is through deep learning models \citep{YUAN2021102221}. The main idea is to derive session representations by making normal sessions deviate from malicious ones in the representation space for deriving anomaly scores.

In many real-world fraud detection scenarios,  only a few labeled malicious and an abundance of normal sessions are available~\citep{YUAN2021102221,DBLP:conf/cikm/YuanZWT20}. These few available malicious sessions usually do not sufficiently cover the entire diversity of all possible malicious sessions. It is well known that malicious sessions can be highly diverse~\citep{YUAN2021102221}. Many attackers keep evolving their activity patterns to avoid detection. Such malicious sessions are usually not  available for training a deep learning model. Suppose a deep learning model is trained by utilizing a few available  malicious sessions. Now in the testing phase, due to the large diversity in the possible malicious sessions, the test set distribution might be different from the training set distribution. For example, the training set might contain only a few types of malicious sessions, and the test set might include other types of malicious sessions that are not observed in the training set. Hence, the learned session representations by using these few malicious sessions in the training set might not be discriminative enough to achieve good generalization on detecting unseen malicious sessions. Clearly, the fraud detection task is essentially an \textit{open-set} detection task.

The existing deep anomaly detection approaches which operate on the setting of a few available anomalous samples, employ metric learning~\citep{pmlr-v80-ruff18a,DBLP:conf/iclr/RuffVGBMMK20} or deviation loss based learning~\citep{DBLP:journals/corr/abs-2108-00462,10.1145/3292500.3330871}. These approaches attempt to obtain a decision boundary by using a few available anomalies. However, these approaches can easily overfit w.r.t seen anomalies and can suffer from poor generalization performance if the anomalies encountered during the testing stage deviate from the training set anomalies~\citep{9879727}. To address this challenge, recently,~\citet{9879727} presented a novel open-set deep anomaly detection approach. They train their model to detect unseen anomalies by jointly employing: (1) a data augmentation strategy through which they generate augmented samples that can closely resemble unseen anomalies, and (2) learning in the latent residual representation space. However, their approach has been specifically designed to operate on image data. In the fraud detection domain, we have additional challenges when compared to the image domain. For example in image data, the normal samples are assumed to have shared features. However, in the fraud detection domain, even normal sessions can also exhibit large diversity. Therefore, learning separable representations for the open-set fraud detection task is challenging.

We address these challenges by leveraging contrastive learning. The vanilla contrastive learning model operates in a self-supervised format. The main goal is to push a sample and its augmented versions closer and contrast with other samples and their corresponding augmented versions in the representation space. However, the employed augmentation strategies are constrained to produce augmented samples that are closely similar to their original versions. Due to this constraint, we cannot employ strong augmentation strategies to generate diverse malicious sessions. Recently,~\citet{DBLP:conf/nips/KhoslaTWSTIMLK20} presented a supervised contrastive learning model which extends contrastive learning to the supervised setting. The main goal is to push the samples belonging to the same class together and contrast with other class samples in the representation space. Due to this class-specific clustering effect in the representation space, the session diversity challenge in our fraud detection task can be effectively addressed. Hence, we leverage this supervised contrastive learning model to build our new robust fraud detection framework called \textit{ConRo}. 

However, the challenge here is to generate those augmented sessions which are similar and can be effective replacements for unseen malicious sessions. ConRo addresses this challenge by employing a two-stage training framework. In the first stage, ConRo trains the session encoder by using the available training set which contains a few malicious and a large amount of normal sessions. Specifically, it performs first stage training by employing a combination of both supervised contrastive and Deep Support Vector
Data Description (DeepSVDD)~\citep{pmlr-v80-ruff18a} losses. Through the supervised contrastive loss, ConRo learns shared features for normal sessions in the representation space, and through  DeepSVDD loss, it pushes normal sessions in a minimum volume hyper-sphere in the representation space. After this stage, ConRo creates a representation space with suitable topological properties which aid in generating potentially diverse malicious sessions. In the second stage, by employing suitable augmentation strategies, ConRo generates diverse potential malicious sessions in the representation space, filters those generated sessions which are false positives/normal, and further trains the encoder through supervised contrastive loss by employing available and generated potential malicious sessions. We summarize our main contributions below:

\begin{itemize}
    \item We propose a novel framework called ConRo which is specifically designed for the open-set fraud detection task. Our ConRo framework operates in a scenario where only a few malicious and a large amount of normal sessions are available. 
    \item We propose a Long-Short Term Memory (LSTM) based session encoder which is trained by employing both supervised contrastive and DeepSVDD losses.   
    \item We propose a data augmentation strategy to generate diverse potential malicious sessions in the representation space. We propose a strategy to filter generated false positive sessions. 
    \item We theoretically analyze the generalization performance of our ConRo framework and highlight important factors influencing its performance. We present an empirical study on three benchmark fraud detection datasets: CERT~\citep{DBLP:conf/sp/GlasserL13}, UMD-Wikipedia~\citep{10.1145/2783258.2783367}, and Open-stack~\citep{DBLP:conf/ccs/Du0ZS17} in which, we show superior performance of our ConRo framework over state-of-the-art baselines.  
\end{itemize}

\section{Related Work}

{\bf \noindent Anomaly Detection.}
Anomaly detection is to detect data that significantly deviate from the majority of data ~\citep{10.1145/3439950}. Recently, many deep anomaly detection approaches are developed by leveraging deep neural networks to learn representations of data so that, anomalies can be easily differentiated from the normal samples ~\citep{10.1145/3439950,ruff2021unifying}. One common setting of anomaly detection assumes the availability of normal samples and aims to learn a decision boundary based on the normal data distribution~\citep{pmlr-v80-ruff18a, DBLP:conf/iclr/RuffVGBMMK20, DBLP:conf/kdd/PangCCL18}.~\citet{10.1145/3292500.3330871,DBLP:journals/corr/abs-2108-00462} proposed an end-to-end anomaly detection framework called \textit{deviation network} which combines representation learning with anomaly scoring. However, all these deep learning-based anomaly detection  approaches have been designed for closed-set anomaly detection task. In our empirical analysis study, we select some of these  approaches~\citep{pmlr-v80-ruff18a,DBLP:conf/iclr/RuffVGBMMK20,10.1145/3292500.3330871,DBLP:journals/corr/abs-2108-00462} as baselines, and show that they fail to deliver noticeable results on open-set fraud detection task.   

\noindent{\textbf{Insider Threat Detection}}. It is a specific case  of fraud detection wherein, the frauds are committed by organizational insiders. 
Deep learning based approaches have become popular in detecting insider threats. We direct the interested readers to~\citet{YUAN2021102221} for a comprehensive survey on deep learning based insider threat detection approaches.  All these deep learning based approaches have not specifically addressed the dataset imbalance challenge in detecting insider threats. Recently, many deep learning based approaches~\citep{DBLP:conf/cikm/YuanZWT20, DBLP:conf/dasfaa/VinayYW22,9776049,DBLP:journals/corr/abs-2103-04475,9789917}
have specifically addressed the dataset imbalance challenge.  However, all these approaches address the closed-set fraud detection task. 

{\bf \noindent Open-Set Recognition (OSR).}~\citet{DBLP:journals/corr/abs-2110-14051} have extensively discussed about contemporary OSR approaches.~\citet{DBLP:conf/kdd/PangHSC21} have addressed anomaly detection in the OSR scenario by employing unlabelled samples.~\citet{DBLP:conf/iclr/HendrycksMD19} trained their model by exposing it to a small set of unseen class samples. However, in our fraud detection task, we do not utilize unlabelled~\citep{DBLP:conf/kdd/PangHSC21} or unseen~\citep{DBLP:conf/iclr/HendrycksMD19} malicious sessions for learning.~\citet{9879727} proposed an open-set anomaly detection framework that learns disentangled representations for different groups of anomalies.  In our empirical study, we select this open-set anomaly detection approach~\citep{9879727} as a baseline and show that it under-performs on open-set fraud detection task. Recently,~\citet{10.1145/3580305.3599302} proposed an open-set anomaly detection framework which operates in the semi-supervised setting. However, our fraud detection task does not operate in the semi-supervised setting.

\noindent{\textbf{Contrastive Learning.}}
~\citet{DBLP:journals/corr/abs-2011-00362} presented an in-depth discussions on the applications of self-supervised contrastive learning on computer vision and NLP domains. In the literature, there is no work studying benefits of supervised contrastive learning for the open-set fraud detection task.

\section{ConRo Framework}

The user activities are modeled through activity sessions. Each session can consist of $T$ user activities. Let $e_{i_t}(1\leq t\leq T)$ denote the $t^{th}$ activity of the $i^{th}$ session. Each activity in a session is represented by an embedding vector, which can be trained based on the word-to-vector model. Let $\mathbf{x}_{i_{t}}\in \mathbb{R}^{d}$ denote the word-to-vector representation of activity $e_{i_{t}}$, where $d$ denotes the number of representation dimensions. Here, $\mathbf{x}_i=\{\mathbf{x}_{i_t}\}_{t=1}^{T}$ denotes the \textit{raw representation} of the $i^{th}$ session. Let $\mathcal{X}$ and $\mathcal{Y}$ denote the raw input representation of sessions and label set, respectively. Here, $\mathcal{Y}=\{0,1\}$ where $y=0$ and $y=1$ denote normal and malicious sessions, respectively. Let $\mathcal{D}$ denote the test set distribution over $\mathcal{X}\times \mathcal{Y}$ wherein, the test samples are drawn from $\mathcal{D}$. The training set $\mathcal{T}$ contains a large amount of normal and a few malicious sessions. Let $\mathcal{T}^0$ and $\mathcal{T}^1$ denote sets of normal and malicious sessions in $\mathcal{T}$, respectively. The malicious sessions sampled from $\mathcal{D}$ will also contain those unseen malicious sessions which are not present in $\mathcal{T}^1$.  
Our ConRo framework has an encoder network that maps a session from its raw representation $\mathbf{x}$ to an \textit{encoded representation} vector $\mathbf{z}$. We adopt LSTM as the foundation of our encoder to derive the encoded session representations. Our encoder consists of two hidden layers with the same dimensions. The hidden representations derived from the top layer of LSTM for the activities in the session $\mathbf{x}_i$ are denoted as $\{\mathbf{h}_{i_t}\}_{t=1}^{T}$. Here, $\mathbf{h}_{i_t}\in \mathbb{R}^{d}$. Then, the encoded session representation $\mathbf{z}_i\in \mathbb{R}^{d}$ is computed as
$\mathbf{z}_i= \frac{1}{T}\sum_{t=1}^T\mathbf{h}_{i_t}$.

The main challenge which we are addressing is to design  a procedure to obtain malicious sessions which are sampled from the test set distribution $\mathcal{D}$. To address this challenge, we construct potential malicious sessions which can be similar to malicious sessions sampled from $\mathcal{D}$. There are two main objectives for generating these potential malicious sessions:

\begin{enumerate}
  \item \textbf{MO1.} Malicious sessions usually form multiple clusters in the encoded representation space~\citep{JU2020167}. Malicious sessions belonging to the same cluster usually share close similarities. Hence, we need to generate potential malicious sessions which are similar to a seen malicious session $\mathbf{x}_i\in \mathcal{T}^1$.
  \item \textbf{MO2.} Suppose there are $K$ malicious session clusters. However, the training set $\mathcal{T}$ might only contain $N(N<K)$ session clusters, and sessions belonging to remaining $K-N$ clusters are not present in $\mathcal{T}$. Note that the malicious sessions from these $K-N$ unseen clusters can diverge significantly from seen malicious sessions. We need to generate potential malicious sessions which belong to those $K-N$ clusters to effectively train our encoder. 
\end{enumerate}

ConRo achieves these main objectives by employing a two stage encoder training procedure. An illustration of this training procedure is shown in Figure \ref{fig:illustration_conro}.  We provide detailed descriptions of both these stages below.

\begin{figure*}[htbp]
\centering
\begin{subfigure}{.35\textwidth}
  \centering
  \includegraphics[width=0.95\linewidth]{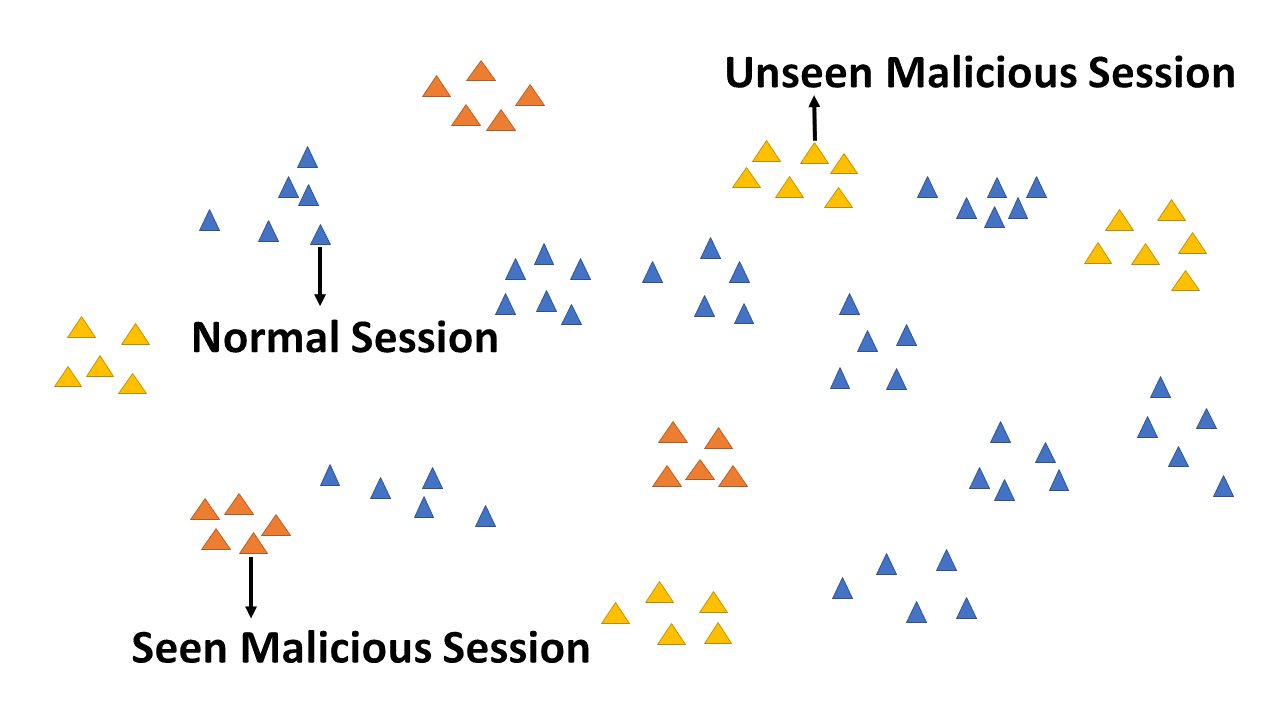}
  \caption{Before training}
  \label{fig:conro_illustration_raw}
\end{subfigure}%
\begin{subfigure}{.35\textwidth}
  \centering
  \includegraphics[width=0.95\linewidth]{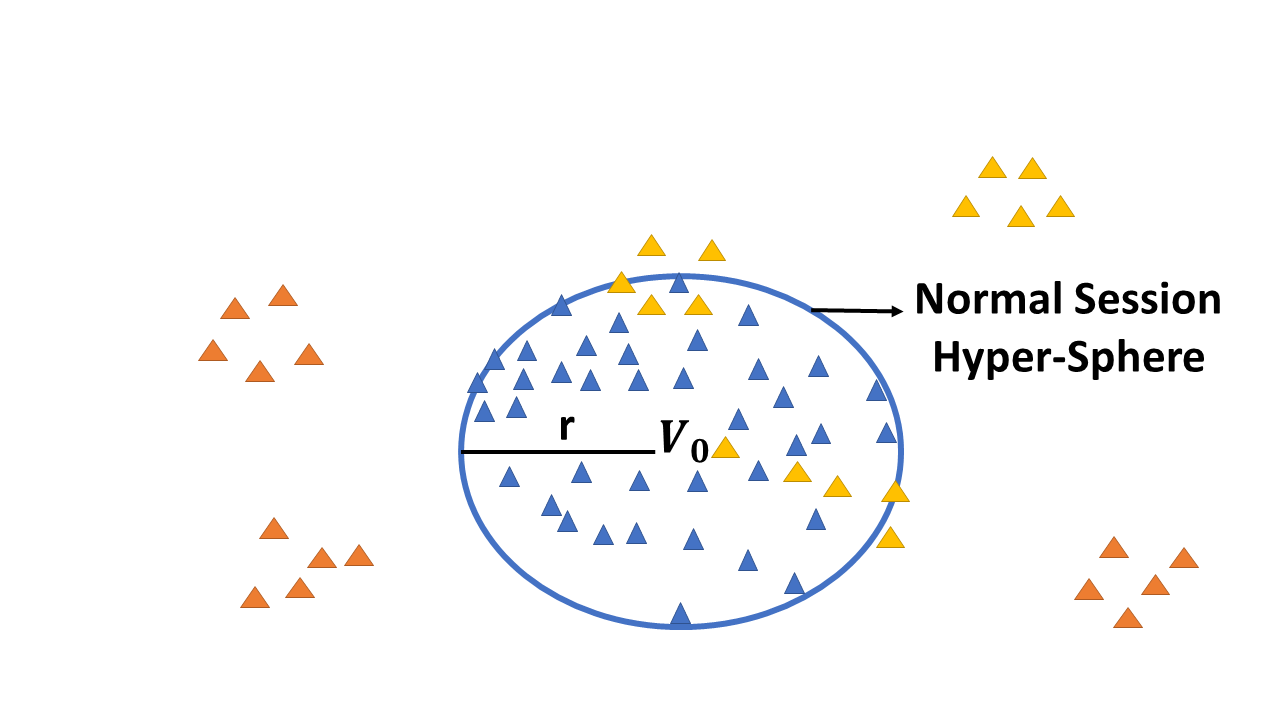}
  \caption{After first stage training}
  \label{fig:conro_illustration_after_stage1_training}
\end{subfigure}
\begin{subfigure}{.35\textwidth}
  \centering
  \includegraphics[width=0.95\linewidth]{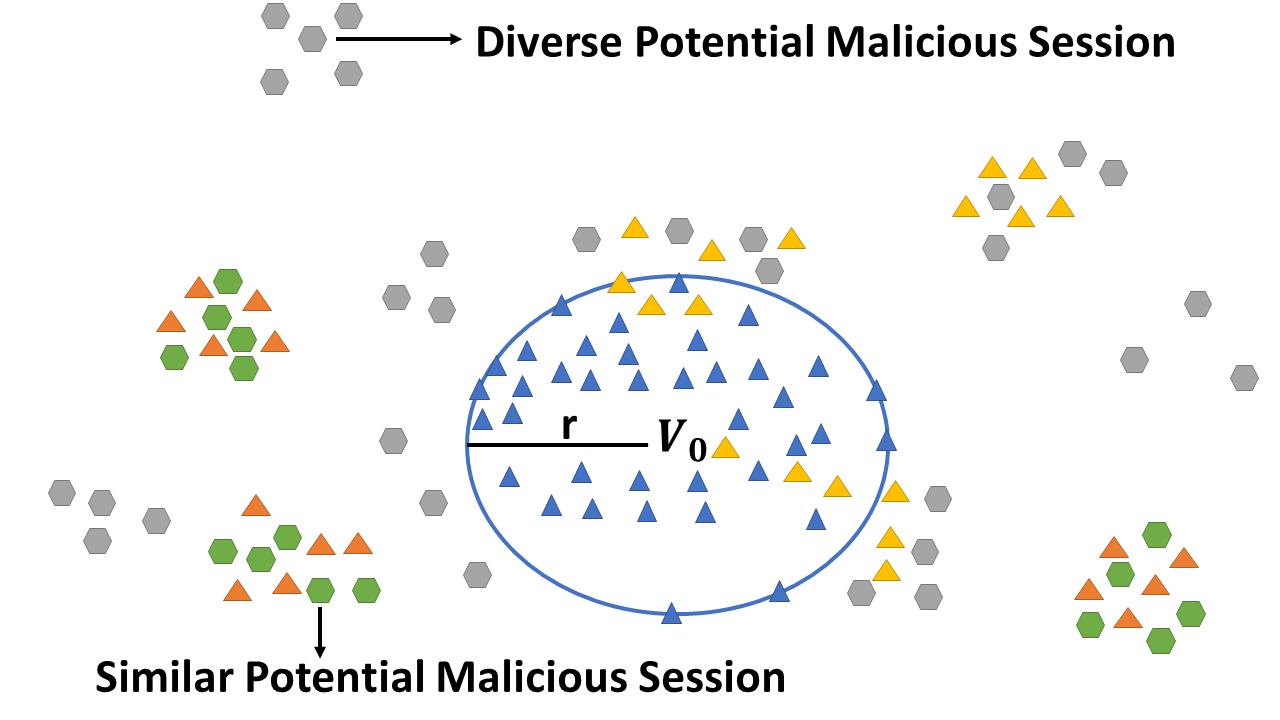}
  \caption{Potential malicious session generation}
  \label{fig:conro_illustration_session_generation}
\end{subfigure}
\begin{subfigure}{.35\textwidth}
  \centering
  \includegraphics[width=0.95\linewidth]{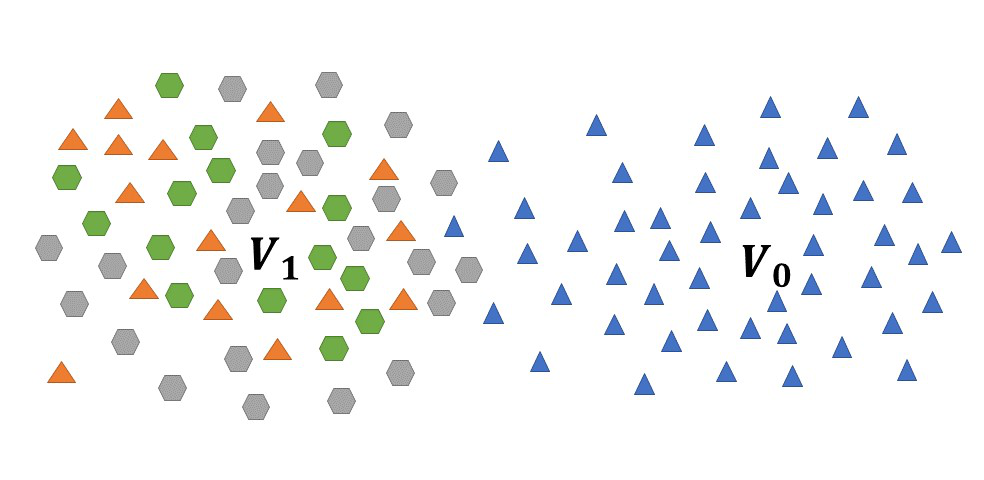}
  \caption{After second stage training}
  \label{fig:conro_illustration_after_stage2_training}
\end{subfigure}
\caption{Illustration of encoded representation space for ConRo.}
\label{fig:illustration_conro}
\end{figure*}

\subsection{First Stage}
In the first stage, our encoder achieves two goals: (1) It learns
shared features for normal sessions and learns to contrast normal sessions with seen malicious sessions in the encoded representation space. As a consequence, our encoder learns separable representations w.r.t to normal and seen malicious sessions. (2) It compresses normal session representations inside a minimum volume hyper-sphere in the encoded representation
space. To achieve the first goal, we leverage the idea of supervised contrastive learning, which can learn separable representations w.r.t to normal and seen malicious sessions. Then, to achieve the second goal, we leverage the DeepSVDD loss~\citep{pmlr-v80-ruff18a}, which pushes the normal samples inside a minimum volume hyper-sphere.

\noindent{\textbf{Supervised contrastive loss.}}  We construct a training batch denoted as $S=\{\mathbf{x}_i\}_{i=1}^R$ by obtaining $R$ random samples from $\mathcal{T}$. Since ConRo is specifically designed to operate on imbalanced training data, in order to effectively contrast malicious sessions with normal sessions, for each training batch $S$, we create a corresponding auxiliary batch $S^1=\{\mathbf{x}_i^1\}_{i=1}^M$, by randomly sampling $M$ malicious sessions from $\mathcal{T}^1$. We leverage a supervised contrastive loss function similar to the one presented by~\citet{DBLP:conf/nips/KhoslaTWSTIMLK20}. This loss function is given by:

\begin{align}
\label{eq:variant2_stage1_supervised_contrastive_loss}
\mathcal{L}^{Sup}=\frac{1}{R}
\sum_{i=1}^{R}
\left(1-y_i\right)\left(\frac{1}{|B^0(\mathbf{x}_i)|}\sum_{\mathbf{x}_p\in B^0(\mathbf{x}_i)}l\left(\mathbf{z}_i,\mathbf{z}_p,A(\mathbf{x}_i)\right)\right)
\end{align}

Here, the set $A(\mathbf{x}_i)$ is defined as $\left(S\cup S^1\right)-\{\mathbf{x}_i\}$,
and the set $B^0(\mathbf{x}_i)=\left\{\mathbf{x}_p\in A(\mathbf{x}_i)|y_p=0\right\}$ indicates samples $\mathbf{x}_p$  in $A(\mathbf{x}_i)$ with labels $y_p=0$. The individual loss $l\left(\mathbf{z}_i,\mathbf{z}_p,A(\mathbf{x}_i)\right)$ between the pair $(\mathbf{x}_i,\mathbf{x}_p)$ is defined as:

\begin{equation}
\label{eq:pair_supervised_contrastive_loss}
l\left(\mathbf{z}_i,\mathbf{z}_p,A(\mathbf{x}_i)\right)=-\log\left(\frac{exp(\cos{\left(\mathbf{z}_i\cdot \mathbf{z}_p\right)}/\alpha)}{\sum_{\mathbf{x}_j\in A(\mathbf{x}_i)}exp(\cos{\left(\mathbf{z}_i\cdot \mathbf{z}_j\right)}/\alpha)}\right),
\end{equation}
where $\alpha$ denotes the temperature parameter.

\noindent{\textbf{DeepSVDD loss.}} We leverage a DeepSVDD loss function which is similar to the one presented by~\citet{pmlr-v80-ruff18a}.  Let $\mathbf{v}_0=\frac{1}{R_0}\sum_{i=1}^{R}(1-y_i)\mathbf{z}_i$ denote the estimated center of normal sessions in the encoded representation space and $R_0=\sum_{i=1}^R\mathbb{I}(y_i=0)$, where $\mathbb{I}(\cdot)$ is an indicator function. This loss function is given by: 

\begin{align}
\label{eq:variant2_stage1_deepsvdd_loss}
\mathcal{L}^{SV} = \frac{1}{R}
\sum_{i=1}^R\left(1-y_i\right)\left(||\mathbf{z}_i-\mathbf{v}_0||_2\right)
\end{align}

The loss function for the first stage is given by:

\begin{equation}
\label{eq:variant2_stage1_loss}
\mathcal{L}_1=\mathcal{L}^{Sup}+\mathcal{L}^{SV}
\end{equation}

To effectively address the session diversity challenge, we employ an alternating approach to optimize our encoder through $\mathcal{L}_1$, instead of joint optimization. In our ablation analysis study described in Section \ref{sec:ablation_analysis}, we show that the alternating optimization approach provides significant performance improvement over the joint optimization approach. For each training batch $S=\{\mathbf{x}_i\}_{i=1}^R$, we first train our encoder through $\mathcal{L}^{Sup}$. As a result, we force the encoder to learn shared features for normal sessions and contrast with seen malicious sessions in the encoded representation space (goal 1). Then, by using the same batch $S$, we train the encoder through $\mathcal{L}^{SV}$, which forces the encoder to compress normal session representations inside a minimum volume hyper-sphere in the encoded representation space (goal 2). 

After the first stage, an encoded representation space having certain topological properties is created (refer to Figures \ref{fig:conro_illustration_raw} and \ref{fig:conro_illustration_after_stage1_training}). For example, in the encoded representation space, most of the normal sessions are pushed inside a minimum volume hyper-sphere.  The seen malicious sessions are found outside this hyper-sphere in multiple clusters and pulled apart from the normal session hyper-sphere. An attractive option now is to directly deploy our first stage trained encoder for test case inference wherein, a test case session $\mathbf{x}$ is predicted as malicious if $||\mathbf{z}-\mathbf{v}_0||>r$ where $r$ is the radius of the normal session hyper-sphere. Otherwise, it is predicted as normal. Note that after the first stage, normal sessions have been contrasted with only seen malicious sessions belonging to $\mathcal{T}^1$, and not unseen malicious sessions. Thus, it is possible that many unseen malicious session clusters overlap with the normal session hyper-sphere, especially those which are similar to normal sessions (refer to Figure \ref{fig:conro_illustration_after_stage1_training}). Hence, by directly deploying the first stage trained encoder for test case inference might negatively affect generalization performance. In our ablation analysis study, we show that this option provides sub-optimal results. Hence, we require a strategy to generate diverse potential malicious sessions which can be similar to unseen malicious sessions in the encoded representation space, further train our encoder by employing these sessions, and learn separable representations w.r.t to the open-set fraud detection task. 

\subsection{Second Stage}

In the second stage, for each seen malicious session $\mathbf{x}_i\in\mathcal{T}^1$, we achieve the first objective (MO1) by generating potential/augmented malicious sessions closely resembling $\mathbf{x}_i$. The first objective can be achieved by leveraging existing augmentation techniques~\citep{DBLP:conf/icml/VermaLKPL21}. However, the second objective (MO2) cannot be achieved straightforwardly due to the lack of required augmentation strategies in the literature~\citep{DBLP:journals/corr/abs-2011-00362}. We achieve the second objective by using an effective augmentation strategy to generate a large amount of potential malicious sessions that span diverse regions of the encoded representation space (refer to Figure \ref{fig:conro_illustration_session_generation}). Note that many of these generated sessions might fall inside the normal session hyper-sphere. In such scenario, we have two choices: (1) a pessimistic choice where we consider such sessions as potential malicious sessions, and (2) an optimistic choice where we consider such sessions as false positives, and filter these sessions. Since deep learning models are sensitive to label noise, choosing the pessimistic choice could result in reduced generalization performance. Hence, we opt for the optimistic choice. In our ablation analysis study, we show that making the pessimistic choice yields sub-optimal results. We design a session filtering mechanism to filter such generated sessions which fall inside the normal session hyper-sphere. Since we generate a large amount of diverse potential malicious sessions, many of them can be located just outside the boundary of the normal session hyper-sphere, which could partially cover unseen malicious session clusters that overlap this hyper-sphere and aid in learning separable representations w.r.t to open-set fraud detection task. We do not individually train our encoder by considering normal sessions from $\mathcal{T}$ to avoid over-fitting.

\noindent{\textbf{First objective}}. We generate \textit{similar potential malicious sessions} which are similar to a seen malicious session $\mathbf{x}_i\in\mathcal{T}^1$. Recently,~\citet{DBLP:conf/icml/VermaLKPL21} proposed a mix-up based data augmentation strategy for sequential data. Their augmentation strategy is inspired by the concept of \textit{convex sets} and generates augmented samples that are similar to their original version. Specifically, they generate augmented samples by performing a mix-up operation on the encoded representations of original samples. Hence, we leverage a mix-up based augmentation strategy which is similar to the one presented by~\citet{DBLP:conf/icml/VermaLKPL21} for generating similar potential malicious sessions which are similar to a seen malicious session $\mathbf{x}_i\in \mathcal{T}^1$. 

Let $\widehat{G}^1(\mathbf{x}_i)$ denote this set of generated similar potential malicious sessions. The set $\widehat{G}^1(\mathbf{x}_i)$ is defined  as $\widehat{G}^1(\mathbf{x}_i)=\{\widehat{\mathbf{z}}|\widehat{\mathbf{z}}=\lambda_1 \mathbf{z}_i+(1-\lambda_1)\mathbf{z}_j\}$. Here, $\widehat{\mathbf{z}}$ denotes the encoded representation of a generated similar potential malicious session, $\mathbf{x}_j\in B^1(\mathbf{x}_i)=\left\{\mathbf{x}_p\in A(\mathbf{x}_i)|y_p=1\right\}$ indicates samples $\mathbf{x}_p$  in $A(\mathbf{x}_i)$ with labels $y_p=1$, $\lambda_1$ is sampled from the Uniform distribution $U(\beta_1,1)$ where $\beta_1 \in [0,1]$, and $\beta_1$ is set closer to 1 to ensure that generated potential malicious sessions have close similarities with $\mathbf{x}_i$. 

\textit{Intuitions behind the design of $\widehat{G}^1(\mathbf{x}_i)$}. Since $\beta_1\in [0,1]$ and $\beta_1$ is closer to 1, and  $\lambda_1$ is sampled from the Uniform distribution $U(\beta_1,1)$, by using $\lambda_1 \mathbf{z}_i$, we get a potential malicious session which is similar to and in the same direction of $\mathbf{z}_i$.  Now, we want to generate potential malicious sessions which are not just confined to the same direction of $\mathbf{z}_i$ and are surrounding $\mathbf{z}_i$ in different directions. Thus, the operation $\widehat{\mathbf{z}}=\lambda_1 \mathbf{z}_i+(1-\lambda_1)\mathbf{z}_j$, aids in achieving this goal. Additionally, performing mix-up with malicious session $\mathbf{x}_j$ will aid in learning separable representations. 

\noindent{\textbf{Second objective}}. 
We generate \textit{diverse potential malicious sessions} which can diverge significantly from  a seen malicious session $\mathbf{x}_i\in\mathcal{T}^1$.  Let $\widetilde{G}^1(\mathbf{x}_i)$ denote this set of generated diverse potential malicious sessions. Our session augmentation strategy is inspired by the concept of \textit{affine sets}.  The set $\widetilde{G}^1(\mathbf{x}_i)$ is defined  as $\widetilde{G}^1(\mathbf{x}_i)=\{\widetilde{\mathbf{z}}|\widetilde{\mathbf{z}}=\lambda_2 \mathbf{z}_i+(1-\lambda_2)\mathbf{z}_j, fp(\widetilde{\mathbf{z}})=0\}$. Here, $\widetilde{\mathbf{z}}$ denotes the encoded representation of a generated diverse potential malicious session, $\lambda_2\sim U(-\beta_2,\beta_2)$, and $\beta_2\in \mathbb{R}$. We treat $\beta_2$ as a hyper-parameter in our empirical studies. We filter a false positive through the function $fp(\cdot)$ as:  

\begin{equation}
\label{eq:false_positive_session}
fp(\widetilde{\mathbf{z}})=   
\begin{cases}
1, \quad \text{if } ||\widetilde{\mathbf{z}}-\mathbf{v}_0||_2\leq r\\
0, \quad{otherwise}
\end{cases}
\end{equation}

\begin{figure*}
\centering
\boxed{\includegraphics[width=120mm,height=58mm]{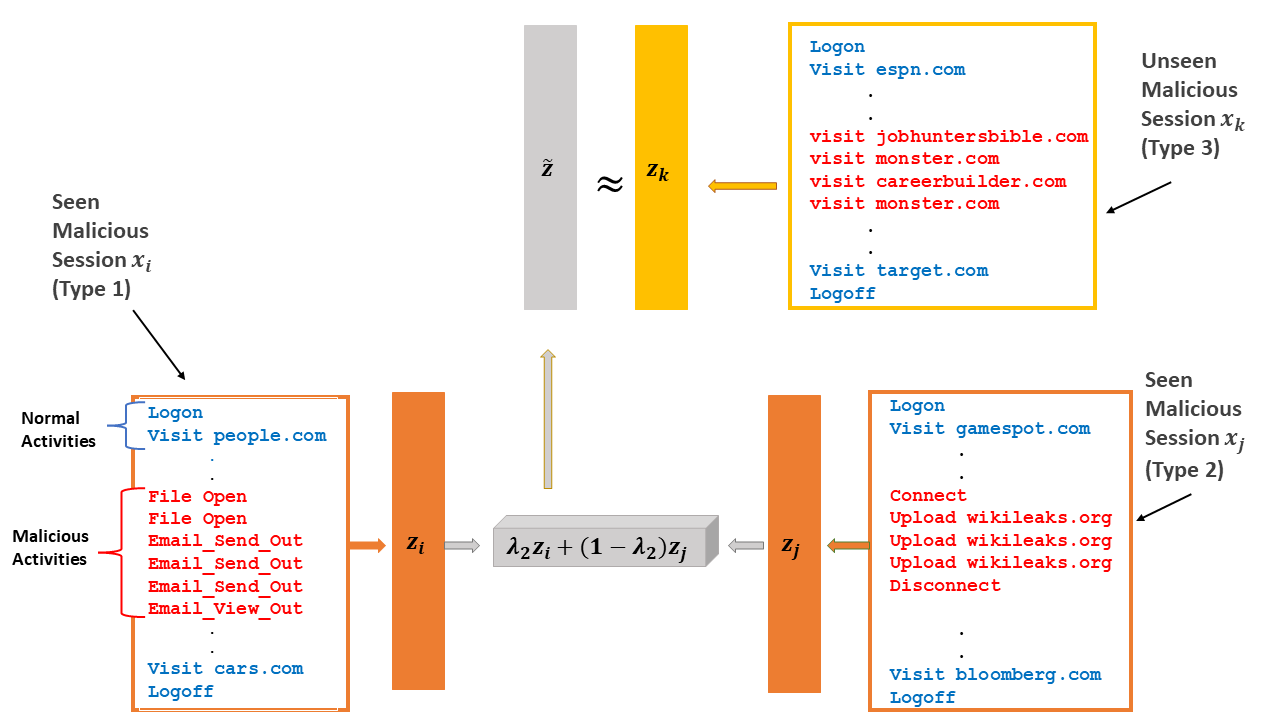}}
\caption{Illustration of generating a diverse potential malicious session (CERT dataset).}
\label{fig:diverse_potential_malicious_session}
\end{figure*}

\textit{Intuitions behind the design of $\widetilde{G}^1(\mathbf{x}_i)$}. We describe our intuitions by using the CERT insider threat dataset. In this dataset, there are five types of malicious sessions. For the ease of description, we consider only three types. Type-1 sessions are related to frauds involving stealing sensitive information and emailing it to a malicious agent. Type-2 sessions are related to devising frauds which involve misusing hardware devices such as removable drives. Type-3 sessions are related to visiting unethical websites such as job-portals with an intent to abandon the current organization. Consider two malicious sessions $\mathbf{x}_i$ and $\mathbf{x}_j$ which belong to type-1 and type-2, respectively. These sessions are illustrated in Figure \ref{fig:diverse_potential_malicious_session}.  Note that both $\mathbf{x}_i$ and $\mathbf{x}_j$, follow a similar activity sequence pattern wherein, normal activities are followed by malicious activities, and finally again followed by normal activities.  

For the malicious session $\mathbf{x}_i$, we can express $\mathbf{z}_i$ as  $\mathbf{z}_i=\mathbf{z}^0_i  \cup \mathbf{z}^1_i$ where $\mathbf{z}^0_i$ and $\mathbf{z}^1_i$ denote  encoded representation feature sets corresponding to normal and malicious activities, respectively. Let $\mathbf{z}^0_i\in \mathbb{R}^n$, $\mathbf{z}^1_i\in \mathbb{R}^m$, and $n+m=d$. Since sessions $\mathbf{x}_i$ and $\mathbf{x}_j$ follow a similar activity sequence pattern, we can hypothesize that $\mathbf{z}_j=\mathbf{z}^0_j  \cup \mathbf{z}^1_j$ with $\mathbf{z}^0_j\in \mathbb{R}^n$ and $\mathbf{z}^1_j\in \mathbb{R}^m$.  Due to the effect of first stage training, normal activity features are tightly clustered in the encoded representation space. Hence, we can infer that $\mathbf{z}^0_i\approx \mathbf{z}^0_j$. Consider an unseen malicious session $\mathbf{x}_k$ belonging to type-3 which is not present in $\mathcal{T}^1$. This session $\mathbf{x}_k$ is shown in Figure \ref{fig:diverse_potential_malicious_session}. Clearly, $\mathbf{x}_k$ also follows a similar activity sequence pattern as $\mathbf{x}_i$ and $\mathbf{x}_j$.  Assume that $\mathcal{T}^1$ does not contain any type-3 malicious sessions. Now we will describe how $\widetilde{G}^1(\mathbf{x}_i)$ can aid in approximating unseen malicious sessions such as $\mathbf{x}_k$. Since $\mathbf{x}_k$ follows a similar activity sequence pattern as $\mathbf{x}_i$ and $\mathbf{x}_j$, we can hypothesize that $\mathbf{z}_k=\mathbf{z}^0_k  \cup \mathbf{z}^1_k$ with $\mathbf{z}^0_k\in \mathbb{R}^n$ and $\mathbf{z}^1_k\in \mathbb{R}^m$. We can infer the result $\mathbf{z}^0_i\approx \mathbf{z}^0_j\approx \mathbf{z}^0_k$. Then for any $\lambda_2\in \mathbb{R}$, we have that: $\mathbf{z}^0_k\approx \lambda_2\mathbf{z}^0_i+ (1-\lambda_2)\mathbf{z}^0_j$. Note that $\mathbf{x}_i$, $\mathbf{x}_j$ and $\mathbf{x}_k$ belong to different malicious session types and hence, $\mathbf{z}^1_i$, $\mathbf{z}^1_j$, and $\mathbf{z}^1_k$ can diverge significantly from each other. By sampling a suitable value for $\lambda_2\sim U(-\beta_2,\beta_2)$, we employ $\lambda_2$ as a coefficient, and aim to obtain the result: $\mathbf{z}^1_k\approx \lambda_2\mathbf{z}^1_i+ (1-\lambda_2)\mathbf{z}^1_j$.

\noindent{\textbf{Second stage loss.}} We again leverage supervised contrastive loss to design our second stage loss function which is given by: 

\begin{align}
\label{eq:variant2_stage2_loss}
\mathcal{L}_2
=\frac{1}{R}\sum_{i=1}^R\Bigg[y_i\Bigg(\frac{1}{|D(\mathbf{x}_i)|}
\sum_{\mathbf{x}_p\in D(\mathbf{x}_i)}l\left(\mathbf{z}_i,\mathbf{z}_p,C(\mathbf{x}_i)\right)\Bigg)\Bigg]
\end{align}

Here, $C(\mathbf{x}_i)=A(\mathbf{x}_i)\cup \widehat{G}^1(\mathbf{x}_i)\cup\widetilde{G}^1(\mathbf{x}_i)$,  $D(\mathbf{x}_i)=B^1(\mathbf{x}_i)\cup\widehat{G}^1(\mathbf{x}_i)\cup\widetilde{G}^1(\mathbf{x}_i)$ and $l\left(\mathbf{z}_i,\mathbf{z}_p,C(\mathbf{x}_i)\right)$ denotes the individual loss between the pair $(\mathbf{x}_i,\mathbf{x}_p)$ corresponding to the malicious sessions defined as:

\begin{equation}
\label{eq:function_l_stage2}
l\left(\mathbf{z}_i,\mathbf{z}_p,C(\mathbf{x}_i)\right)=-\log\left(\frac{exp(\cos{\left(\mathbf{z}_i\cdot \mathbf{z}_p\right)}/\alpha)}{\sum_{\mathbf{x}_j\in C(\mathbf{x}_i)}exp(\cos{\left(\mathbf{z}_i\cdot \mathbf{z}_j\right)}/\alpha)}\right)
\end{equation}

\subsection{Training}

\begin{algorithm}
\caption{Training procedure for ConRo.} 
\begin{algorithmic}[1]
\Statex \textbf{Inputs}: $\mathcal{T}=\mathcal{T}^1\cup\mathcal{T}^0$, $R$, $M$, $\widehat{M}$,  $\widetilde{M}$, $\beta_1$, $\beta_2$ and our untrained encoder.
\Statex \textbf{Output}: well trained encoder.
\State generate raw representations of all the sessions in $\mathcal{T}$;
\Statex [\textbf{First Stage}]
\For{each training batch $S=\{\mathbf{x}_i\}_{i=1}^R$ generated from $\mathcal{T}$}
\State create the auxiliary batch $S^1=\{\mathbf{x}_i^1\}_{i=1}^M$ from $\mathcal{T}^1$;
\State calculate $R_0=\sum_{i=1}^R\mathbb{I}(y_i=0)$ and $\mathbf{v}_0=\frac{1}{R_0}\sum_{i=1}^{R}(1-y_i)\mathbf{z}_i$;
\For{each normal session $(\mathbf{x}_i,y_i=0)\in S$}
\State construct set $A(\mathbf{x}_i)=\left(S\cup S^1\right)-\{\mathbf{x}_i\}$;
\State construct set $B^{0}(\mathbf{x}_i)=\left\{\mathbf{x}_p\in A(\mathbf{x}_i)|y_p=0\right\}$;
\State calculate $l\left(\mathbf{z}_i,\mathbf{z}_p,A(\mathbf{x}_i)\right)$ for each session $\mathbf{x}_p\in B^{0}(\mathbf{x}_i)$ by using Equation \ref{eq:pair_supervised_contrastive_loss};
\EndFor
\State calculate  $\mathcal{L}^{Sup}$ by using Equation \ref{eq:variant2_stage1_supervised_contrastive_loss} and train the encoder;
\State calculate  $\mathcal{L}^{SV}$ by using Equation \ref{eq:variant2_stage1_deepsvdd_loss} and train the encoder;
\EndFor
\Statex [\textbf{Second Stage}]
\For{each training batch $S=\{\mathbf{x}_i\}_{i=1}^R$ generated from $\mathcal{T}$}
\State create the auxiliary batch $S^1=\{\mathbf{x}_i^1\}_{i=1}^M$ from $\mathcal{T}^1$;
\For{each malicious session $(\mathbf{x}_i,y_i=1)\in S$}
\State construct set $A(\mathbf{x}_i)=\left(S\cup S^1\right)-\{\mathbf{x}_i\}$;
\State construct set $B^{1}(\mathbf{x}_i)=\left\{\mathbf{x}_p\in A(\mathbf{x}_i)|y_p=1\right\}$;
\State $\widehat{G}^1(\mathbf{x}_i)=\phi$, $\widetilde{G}^1(\mathbf{x}_i)=\phi$;
\For{$k=1$ to $\widehat{M}$}
\State sample $\lambda_1$ from $U(\beta_1,1.0)$ and $\mathbf{x}_j$ from $B^{1}(\mathbf{x}_i)$;
\State $\widehat{G}^1(\mathbf{x}_i)=\widehat{G}^1(\mathbf{x}_i)\cup\left\{\lambda_1 \mathbf{z}_i+(1-\lambda_1)\mathbf{z}_j\right\}$;
\EndFor
\For{$k=1$ to $\widetilde{M}$}
\State sample $\lambda_2$ from $U(-\beta_2,\beta_2)$ and $\mathbf{x}_j$ from $B^{1}(\mathbf{x}_i)$;
\State $\widetilde{\mathbf{z}}=\lambda_2 \mathbf{z}_i+(1-\lambda_2)\mathbf{z}_j$ and calculate $fp(\widetilde{\mathbf{z}})$ by using Equation \ref{eq:false_positive_session};
\If{$fp(\widetilde{\mathbf{z}})$==0}
\State $\widetilde{G}^1(\mathbf{x}_i)=\widetilde{G}^1(\mathbf{x}_i)\cup \widetilde{\mathbf{z}}$;
\EndIf
\EndFor
\State construct set $C(\mathbf{x}_i)=A(\mathbf{x}_i)\cup\widehat{G}^1(\mathbf{x}_i)\cup\widetilde{G}^1(\mathbf{x}_i)$
\State construct set $D(\mathbf{x}_i)=B^1(\mathbf{x}_i)\cup\widehat{G}^1(\mathbf{x}_i)\cup\widetilde{G}^1(\mathbf{x}_i)$
\State calculate $l\left(\mathbf{z}_i,\mathbf{z}_p,C(\mathbf{x}_i)\right)$ for each session $\mathbf{x}_p \in D(\mathbf{x}_i)$ by using Equation \ref{eq:function_l_stage2};
\EndFor
\State calculate $\mathcal{L}_2$ using Equation \ref{eq:variant2_stage2_loss} and train the encoder;
\EndFor
\State \textbf{return} well trained encoder;
\end{algorithmic}
\label{alg:pseudocode_withoutu}
\end{algorithm}

The ConRo training procedure is outlined in Algorithm \ref{alg:pseudocode_withoutu}.  Let $\widehat{M}=|\widehat{G}^1(\mathbf{x}_i)|$ and $\widetilde{M}\geq|\widetilde{G}^1(\mathbf{x}_i)|$. Initially, we generate raw representations for all the sessions in the training set $\mathcal{T}$. In the first stage, we generate training batches from $\mathcal{T}$. For each training batch $S=\{\mathbf{x}_i\}_{i=1}^R$, we create the corresponding auxiliary batch $S^1=\{\mathbf{x}_i^1\}_{i=1}^M$ from $\mathcal{T}^1$. We calculate $R_0=\sum_{i=1}^R\mathbb{I}(y_i=0)$ and $\mathbf{v}_0=\frac{1}{R_0}\sum_{i=1}^{R}(1-y_i)\mathbf{z}_i$ where, $\mathbb{I}(\cdot)$ is an indicator function. Then, for each normal session $\mathbf{x}_i\in S$, we construct the sets $A(\mathbf{x}_i)=\left(S\cup S^1\right)-\{\mathbf{x}_i\}$ and $B^{0}(\mathbf{x}_i)=\left\{\mathbf{x}_p\in A(\mathbf{x}_i)|y_p=0\right\}$. For each session $\mathbf{x}_p \in  B^0(\mathbf{x}_i)$, we calculate $l\left(\mathbf{z}_i,\mathbf{z}_p,A(\mathbf{x}_i)\right)$ by using Equation \ref{eq:pair_supervised_contrastive_loss}. We calculate the batch supervised contrastive loss using  Equation \ref{eq:variant2_stage1_supervised_contrastive_loss} 
and train the encoder by using this batch loss. Then, we calculate the batch DeepSVDD loss using  Equation \ref{eq:variant2_stage1_deepsvdd_loss}
and again train the encoder by using this batch loss.

In the second stage, we again generate training batches from $\mathcal{T}$. For each training batch $S=\{\mathbf{x}_i\}_{i=1}^R$, we create the auxiliary batch $S^1=\{\mathbf{x}_i^1\}_{i=1}^M$ from $\mathcal{T}^1$. For each malicious session $\mathbf{x}_i\in S$, we construct sets $A(\mathbf{x}_i)=\left(S\cup S^1\right)-\{\mathbf{x}_i\}$ and $B^1(\mathbf{x}_i)=\{\mathbf{x}_p\in A(\mathbf{x}_i)|y_p=1\}$. We construct the set $\widehat{G}^1(\mathbf{x}_i)$ consisting of $\widehat{M}$ potential malicious sessions which are similar to $\mathbf{x}_i$. Specifically, we take a sample $\lambda_1$ from $U(\beta_1,1.0)$, we sample a session $\mathbf{x}_j$ from $B^{1}(\mathbf{x}_i)$, and generate a similar potential malicious session by applying the operation $\lambda_1 \mathbf{z}_i+(1-\lambda_1)\mathbf{z}_j$. We create another set $\widetilde{G}^1(\mathbf{x}_i)$ consisting of a maximum of $\widetilde{M}$ diverse potential malicious sessions. Specifically, we take a sample $\lambda_2$ from $U(-\beta_2,\beta_2)$, we sample a session $\mathbf{x}_j$ from $B^{1}(\mathbf{x}_i)$, and generate a diverse potential malicious session $\widetilde{\mathbf{z}}$ by applying the operation $\lambda_2 \mathbf{z}_i+(1-\lambda_2)\mathbf{z}_j$. Then, we calculate $\mathbf{v}_0$ and the radius of the normal session hyper-sphere $r$, and calculate  $fp(\widetilde{\mathbf{z}})$ using Equation \ref{eq:false_positive_session}. If $fp(\widetilde{\mathbf{z}})=0$ then, $\widetilde{\mathbf{z}}$ is considered as a true positive session and we include $\widetilde{\mathbf{z}}$ into $\widetilde{G}^1(\mathbf{x}_i)$. Then, we construct sets $C(\mathbf{x}_i)=A(\mathbf{x}_i)\cup\widehat{G}^1(\mathbf{x}_i)\cup\widetilde{G}^1(\mathbf{x}_i)$ and $D(\mathbf{x}_i)=B^1(\mathbf{x}_i)\cup\widehat{G}^1(\mathbf{x}_i)\cup\widetilde{G}^1(\mathbf{x}_i)$. 
For each session $\mathbf{x}_p \in  D(\mathbf{x}_i)$, we calculate $l\left(\mathbf{z}_i,\mathbf{z}_p,C(\mathbf{x}_i)\right)$ by using Equation \ref{eq:function_l_stage2}. We calculate the batch loss by using $\mathcal{L}_2$ shown in Equation \ref{eq:variant2_stage2_loss} and 
train the encoder by using this batch loss. Finally, the training algorithm returns the well trained encoder.

{\bf \noindent Time complexity analysis}. We analyze the time complexity of our ConRo training procedure by considering the forward pass and the number of times the individual loss $l(\cdot,\cdot,\cdot)$ is invoked in both stages. This time complexity is given by: 
$O\left(|\mathcal{T}^0|R+|\mathcal{T}^1|\left(M+|\widehat{G}^1(\mathbf{x}_i)|+|\widetilde{G}^1(\mathbf{x}_i)|\right)\right)$.

{\bf \noindent Inference.} After the second stage training, our encoder has learnt to push seen malicious sessions, similar and diverse potential malicious sessions closer in the encoded representation space, and as a consequence, all these sessions form a tight cluster in the encoded representation space (refer to Figure \ref{fig:conro_illustration_after_stage2_training}). Normal sessions are also tightly clustered in the encoded representation space due to the effect of first stage training. Hence, we design our inference strategy by analyzing the proximities of a test case session to the centers of normal and malicious sessions in the encoded representation space. Let $\mathbf{v}_1$ denote the estimated center of malicious sessions in the encoded representation space, which is given by $\mathbf{v}_1=\frac{1}{M}\sum_{i=1}^M\mathbf{z}_i^1$, where $\{\mathbf{x}_i^1\}_{i=1}^M$ denotes $M$ randomly sampled malicious sessions from $\mathcal{T}^1$. For any test case session $\mathbf{x}$, ConRo predicts its label as:

$label(\mathbf{x})=\begin{cases}
1 &\text{ if }\text{  } ||\mathbf{z}-\mathbf{v}_1||_2<||\mathbf{z}-\mathbf{v}_0||_2 \\
0 & otherwise
\end{cases}$

\subsection{Theoretical Analysis}

We present a theoretical analysis study to highlight the important factors which influence the generalization performance of our ConRo framework. We introduce a set definition called $\epsilon$-$span(\mathbf{x}_j)$ for a session $\mathbf{x}_j$ which is defined as:

\resizebox{.97\linewidth}{!}{
\begin{minipage}{\linewidth}
\begin{align}
\label{eq:span}
\epsilon\text{-}span(\mathbf{x}_j)=\left\{\mathbf{x}_k\bigg| ||\mathbf{z}_k-\mathbf{z}_j||_2\leq \epsilon, P\left(||\mathbf{v}_y-\mathbf{z}_k||_2<||\mathbf{v}_{1-y}-\mathbf{z}_k||_2\bigg|||\mathbf{v}_y-\mathbf{z}_j||_2<||\mathbf{v}_{1-y}-\mathbf{z}_j||_2\right) = 1 \right\} 
\end{align}
\end{minipage}
}

Here, $y=\{0,1\}$, $\epsilon$-$span(\mathbf{x}_j)$ contains those sessions $\mathbf{x}_k$ such that  $||\mathbf{z}_k-\mathbf{z}_j||_2\leq \epsilon$ for some $\epsilon\geq0$, and our encoder has similar session projection action w.r.t to proximities between $\mathbf{v}_0$ and $\mathbf{v}_1$ for both $\mathbf{x}_k$ and $\mathbf{x}_j$, which is as shown in the right hand side second term of Equation \ref{eq:span}. We extend the definition of $\epsilon$-$span(\mathbf{x}_j)$ to a set of sessions $\mathcal{S}$ called $\epsilon$-$span(\mathcal{S})$. This set definition is given by:

\resizebox{.97\linewidth}{!}{
\begin{minipage}{\linewidth}
\begin{align*}
\epsilon\text{-}span(\mathcal{S})=\left\{\mathbf{x}_k\bigg|\exists \mathbf{x}_j\in \mathcal{S}, ||\mathbf{z}_k-\mathbf{z}_j||_2\leq \epsilon, P\left(||\mathbf{v}_y-\mathbf{z}_k||_2<||\mathbf{v}_{1-y}-\mathbf{z}_k||_2\bigg|||\mathbf{v}_y-\mathbf{z}_j||_2<||\mathbf{v}_{1-y}-\mathbf{z}_j||_2\right) = 1 \right\} 
\end{align*}
\end{minipage}
}

Here, $\epsilon$-$span(\mathcal{S})$ contains those sessions $\mathbf{x}_k$ such that there exists some session $\mathbf{x}_j\in \mathcal{S}$ where $||\mathbf{z}_k-\mathbf{z}_j||_2\leq \epsilon$ for some $\epsilon\geq0$, and our encoder has similar session projection action w.r.t to proximities between $\mathbf{v}_0$ and $\mathbf{v}_1$ for both $\mathbf{x}_k$ and $\mathbf{x}_j$. 
For our theoretical analysis, we assume an existence of a hypothetical oracle version of our encoder which is trained by using the labeled sessions sampled from the test distribution $\mathcal{D}$ and supervised contrastive loss function. Let $\mathbf{v}^o_0$ and $\mathbf{v}_1^o$ denote the centers of normal and malicious sessions in the encoded representation space corresponding to this oracle encoder, respectively. For a test case session $\mathbf{x}$, the oracle encoder has the following properties:

$P\left(||\mathbf{v}_1^o-\mathbf{z}||_2<||\mathbf{v}_{0}^o-\mathbf{z}||_2\right)=    
\begin{cases}
1, \quad \text{if } y=1\\
0, \quad{otherwise}
\end{cases}$

\begin{thm}
\label{theorem_2}
For any test case session $\mathbf{x}$ which is sampled from $\mathcal{D}$, the following bound holds:   
\begin{align*}
P\left(||\mathbf{v}_y-\mathbf{z}||_2<||\mathbf{v}_{1-y}-\mathbf{z}||_2\bigg|||\mathbf{v}_y^o-\mathbf{z}||_2<||\mathbf{v}_{1-y}^o-\mathbf{z}||_2\right)
\geq P\left(\mathbf{x}\in \bigcup_{\mathbf{x}_i\in\mathcal{T}^1}\epsilon\text{-}span\left(\widehat{G}^1(\mathbf{x}_i)\cup\mathbf{x}_i\right)\cup\epsilon\text{-}span\left(\widetilde{G}^1(\mathbf{x}_i)\right) \right)
\end{align*}
\end{thm}

\begin{proof}
\textit{Case when} $y=1$. In this case: $P\left(||\mathbf{v}_1^o-\mathbf{z}||_2<||\mathbf{v}_{0}^o-\mathbf{z}||_2\right)=1$. Now our encoder can either project $\mathbf{x}$ closer to $\mathbf{v}_1$ or $\mathbf{v}_0$. If it projects closer to $\mathbf{v}_1$ which means that $||\mathbf{v}_1-\mathbf{z}||_2<||\mathbf{v}_{0}-\mathbf{z}||_2$. Then, there are two scenarios. In the first scenario, we have that:

\begin{align*}
\mathbf{x}\in \bigcup_{\mathbf{x}_i\in\mathcal{T}^1}\epsilon\text{-}span\left(\widehat{G}^1(\mathbf{x}_i)\cup\mathbf{x}_i\right)\cup\epsilon\text{-}span\left(\widetilde{G}^1(\mathbf{x}_i)\right)
\end{align*}

For a set of sessions $\mathcal{S}$, we have defined $\epsilon\text{-}span\left(\mathcal{S}\right)$ as:

\begin{align*}
\epsilon\text{-}span(\mathcal{S})=\left\{\mathbf{x}_k\bigg|\exists \mathbf{x}_j\in \mathcal{S}, ||\mathbf{z}_k-\mathbf{z}_j||_2\leq \epsilon, P\left(||\mathbf{v}_y-\mathbf{z}_k||_2<||\mathbf{v}_{1-y}-\mathbf{z}_k||_2\bigg|||\mathbf{v}_y-\mathbf{z}_j||_2<||\mathbf{v}_{1-y}-\mathbf{z}_j||_2\right) = 1 \right\} 
\end{align*}

By only considering the first scenario and by the definition of $\epsilon\text{-}span\left(\mathcal{S}\right)$, we have the result:

\begin{align*}
P\left(||\mathbf{v}_1-\mathbf{z}||_2<||\mathbf{v}_{0}-\mathbf{z}||_2\bigg|||\mathbf{v}_1^o-\mathbf{z}||_2<||\mathbf{v}_{0}^o-\mathbf{z}||_2\right)
= P\left(\mathbf{x}\in \bigcup_{\mathbf{x}_i\in\mathcal{T}^1}\epsilon\text{-}span\left(\widehat{G}^1(\mathbf{x}_i)\cup\mathbf{x}_i\right)\cup\epsilon\text{-}span\left(\widetilde{G}^1(\mathbf{x}_i)\right) \right)
\end{align*}

In the alternate scenario, we have that:

\begin{align*}
\mathbf{x}\notin \bigcup_{\mathbf{x}_i\in\mathcal{T}^1}\epsilon\text{-}span\left(\widehat{G}^1(\mathbf{x}_i)\cup\mathbf{x}_i\right)\cup\epsilon\text{-}span\left(\widetilde{G}^1(\mathbf{x}_i)\right)
\end{align*}

By considering both these scenarios, we have the result: 

\begin{align*}
P\left(||\mathbf{v}_1-\mathbf{z}||_2<||\mathbf{v}_{0}-\mathbf{z}||_2\bigg|||\mathbf{v}_1^o-\mathbf{z}||_2<||\mathbf{v}_{0}^o-\mathbf{z}||_2\right)
\geq P\left(\mathbf{x}\in \bigcup_{\mathbf{x}_i\in\mathcal{T}^1}\epsilon\text{-}span\left(\widehat{G}^1(\mathbf{x}_i)\cup\mathbf{x}_i\right)\cup\epsilon\text{-}span\left(\widetilde{G}^1(\mathbf{x}_i)\right) \right)
\end{align*}

\textit{Case when} $y=0$. In this case, $P\left(||\mathbf{v}_0^o-\mathbf{z}||_2\leq||\mathbf{v}_{1}^o-\mathbf{z}||_2\right)=1$. Since  our encoder is trained by using a large number of normal sessions in $\mathcal{T}$ and due to which, we can sufficiently cover the diversity in normal sessions, we have that: 
$P\left(||\mathbf{v}_0-\mathbf{z}||_2<||\mathbf{v}_{1}-\mathbf{z}||_2\bigg|||\mathbf{v}_0^o-\mathbf{z}||_2<||\mathbf{v}_{1}^o-\mathbf{z}||_2\right)\approx1$. Thus, the theorem immediately follows. 

\end{proof}

Theorem \ref{theorem_2} outlines a formal explanation on the generalization performance of ConRo. Clearly, this performance is influenced by  $P\left(\mathbf{x}\in \bigcup_{\mathbf{x}_i\in\mathcal{T}^1}\epsilon\text{-}span\left(\widehat{G}^1(\mathbf{x}_i)\cup\mathbf{x}_i\right)\cup\epsilon\text{-}span\left(\widetilde{G}^1(\mathbf{x}_i)\right) \right)$. This probability value can be increased by: (1) setting a suitable value for $\beta_2$ based on empirical analysis, and (2) setting a large value for $|\widetilde{G}^1(\mathbf{x}_i)|$ ensures that generated diverse potential malicious sessions can span diverse regions of the encoded representation space.

\section{Experiments}

We describe our experimental setup\footnote{Code is available through the hyperlink: \color{blue}\href{https://www.dropbox.com/s/0bnt28jvpj0v2he/ConRo_Code.zip?raw=1}{Code Link}.} including datasets and
baselines used in this paper and then discuss our experimental
results including hyper-parameter sensitivity, visualization,  training latency, and ablation analysis results. 

\subsection{Experimental Setup}

\subsubsection{Datasets}
\label{sec:training_set}
We use three benchmark fraud detection datasets: CERT  \citep{DBLP:conf/sp/GlasserL13}, UMD-Wikipedia \citep{10.1145/2783258.2783367}, and OpenStack~\citep{DBLP:conf/ccs/Du0ZS17}.

{\bf \noindent CERT \citep{DBLP:conf/sp/GlasserL13}.} The CERT dataset is a comprehensive dataset for insider threat detection. There are 48 malicious and 1,581,358  normal sessions. 
The insider sessions are chronologically recorded over 516 days. To avoid extreme training latency, we randomly sample 10,000 normal sessions from the first 460 days, and include them in the training set $\mathcal{T}$. Similarly, we randomly sample 500 normal sessions from 461 to 516 days to construct our test set. 
There are 5 types of malicious sessions. (1) \textit{Logon}: The insider logs on a computer during weekends or on a weekday after work hours. (2) \textit{Email}: The insider sends/views unexpected emails to/from external sources. (3) \textit{HTTP}: The insider uploads/downloads organizational information to/from external malicious websites. (4) \textit{Device}: The insider connects a device such as removable drives during weekends or on a weekday after work hours. (5) \textit{File}: The insider manipulates organizational files with malicious intentions. We construct a biased training set corresponding to malicious sessions wherein, we include device, email, and file malicious session types in the training set and remaining two types in the test set. Specifically, we include 30 and 18 malicious sessions in the training and test sets, respectively.

{\bf \noindent UMD-Wikipedia \citep{10.1145/2783258.2783367}.} This dataset consists of activity sessions of a set of users who have edited the Wikipedia website.  In this dataset, there are 5486 normal and 4627 malicious sessions. We randomly sample 1000 normal sessions to construct the test set and include all the remaining 4486 normal sessions in the training set.  For the malicious sessions, in-order to simulate open-set and imbalanced dataset scenario, we construct the training set  by leveraging and suitably adapting the procedure utilized by~\citet{10.1145/3459637.3482104}, which is described below. We calculate the appropriate number of malicious session clusters $(K)$ in the available malicious sessions by using \textit{silhouette coefficient analysis} \citep{ROUSSEEUW198753}. From our empirical study, we get $K=3$. Then, we randomly sample 70 and 10 malicious sessions from the first and second malicious session clusters, respectively, and include them in the training set. From the remaining malicious sessions, we randomly sample similar number of malicious sessions from each of the 3 clusters to construct the test set which contains 500 malicious sessions. 

{\bf \noindent OpenStack~\citep{DBLP:conf/ccs/Du0ZS17}.}  This dataset records the activity sessions of  users who have used the OpenStack cloud services. In this dataset, there are 244,908 normal and 18,434 malicious sessions. We randomly sample 10,000 and 1000 normal sessions and include them in our training and test sets, respectively. For the malicious sessions, through silhouette coefficient analysis we get $k=12$ malicious session clusters. We randomly sample 50 and 10 malicious sessions from the first and second malicious session clusters, respectively, and include them in the training set. From the remaining malicious sessions, we construct a test set having 120 malicious sessions by randomly sampling equal number of malicious sessions from each of the 12 malicious session clusters.

\subsubsection{Training Details}

By considering a user activity session as a sentence, we train the word-to-vector model \citep{mikolov2013efficient} to derive the activity representation. The minimum activity frequency is set as 1 since every activity is given importance in the design.
To effectively train our session encoder, we set the number of dimensions of the activity and session representations as $d=50$. Since we generate encoded session representation by averaging the output sequence of the LSTM model, we set the hidden layer size of LSTM to 50. The temperature parameter $\alpha$ shown in Equations \ref{eq:pair_supervised_contrastive_loss} and \ref{eq:function_l_stage2} is set to its default value 1. We opt for medium sized training batches in order to avoid extreme memory requirements during encoder training. Specifically, we use 100 sessions $(R)$ in each training batch. We set the size of the malicious session auxiliary batch $(M)$ as 20. The sizes of potential malicious session batches $|\widehat{G}^1(\mathbf{x}_i)|$ and $|\widetilde{G}^1(\mathbf{x}_i)|$ are set as 20 and 200, respectively. For $\beta_1$, we set its value as 0.92 because it is supposed to be closer to 1.  For $\beta_2$, we set its value as 4 in order to generate potential malicious sessions which are sufficiently diverse. Additionally, we perform a sensitivity analysis study on $\beta_2$ which is described in Section \ref{sec:sensitivity_analysis_study}. We use the Adam optimizer~\citep{DBLP:journals/corr/KingmaB14} with a learning rate of 0.005 and we use 10 training epochs for both stages.  We utilize three metrics to measure the anomaly detection performance: $F_1$, False Positive Rate (FPR), and Area Under the Receiver Operating Characteristics Curve (AUC-ROC). We report the mean and standard deviation of performance scores after 5 times of running.

\subsubsection{Baselines}

We compare our ConRo framework with five state-of-the-art baselines which were  specifically designed for anomaly detection: DeepSVDD~\citep{pmlr-v80-ruff18a}, DeepSAD~\citep{DBLP:conf/iclr/RuffVGBMMK20}, DevNet~\citep{10.1145/3292500.3330871,DBLP:journals/corr/abs-2108-00462}, CLDet~\citep{DBLP:conf/dasfaa/VinayYW22}, and Swan~\citep{9879727}. All these baselines operate in the setting where only a few anomalous samples are available.  DeepSVDD, DeepSAD, DevNet, and CLDet have been designed for closed set anomaly detection whereas, Swan has been designed for open-set anomaly detection. Specifically, CLDet is a self-supervised contrastive learning based insider threat detection framework. Except CLDet, the remaining baselines originally operate on image datasets, and employ neural  networks for image data such as CNN~\citep{DBLP:conf/iclr/RuffVGBMMK20}, ResNet-18~\citep{9879727} etc. Hence, they cannot be directly applied for our fraud detection task which operates on sequential data. We replace their neural networks with our LSTM-based session encoder and adapt these baselines to our fraud detection task. We employ the same training set used for our ConRo to train all these baselines. For Swan, the original augmentation technique is image-specific. Hence, we replace it with  the augmentation technique proposed by~\citet{DBLP:conf/icml/VermaLKPL21}. Additionally, the unseen malicious sessions are detected in a residual representation space which is defined as: $\mathbf{v}_0-\mathbf{z}$.

\subsection{Experimental Results}

\subsubsection{Overall Comparison}

\begin{table*}[htbp]
\caption{Performances of our ConRo and baselines (mean\textpm std). The higher the better for F1 and AUC-ROC. The lower the better for FPR.  The best values are bold highlighted. 
}
\label{tb:expr_results}
\resizebox{1.0\textwidth}{!}{
\begin{tabular}{|c|c|c|c|c|c|c|c|c|c|}
\hline
\multirow{2}{*}{\textbf{Models}} & \multicolumn{3}{|c|}{\textbf{CERT}} & \multicolumn{3}{|c|}{\textbf{UMD-Wikipedia}} & \multicolumn{3}{|c|}{\textbf{Open-Stack}}\\ \cline{2-10}
 &  F1 & FPR & AUC-ROC & F1 & FPR & AUC-ROC & F1 & FPR & AUC-ROC \\\hline
 DeepSVDD &  14.67\textpm 4.1 & 14.30\textpm 1.8 &  62.29\textpm 4.8  & 33.23\textpm 1.7 & 44.90 \textpm 1.4  & 46.47\textpm 0.8 & 32.27\textpm 0.7 &  42.10\textpm 1.4 & 79.02\textpm 0.7 \\\hline
 DeepSAD &   24.71\textpm 7.5 & 20.53\textpm 7.1 & 84.17\textpm 3.6  & 56.88\textpm 2.9 & 13.30\textpm 0.2 & 68.35\textpm 1.7  & 67.43\textpm 3.1 &  9.70\textpm 1.3 & 94.12\textpm 0.1 \\\hline
CLDet &  60.41\textpm 3.6  &  3.75\textpm 1.9  &  79.47\textpm 2.6 &  58.78\textpm 3.1 &  9.59\textpm 2.8 &  70.71\textpm 2.4  &  61.71\textpm 2.9   &  6.18\textpm 2.1   &  83.98\textpm 1.8 \\\hline
 Swan &   59.31\textpm 2.2 & \textbf{0.0\textpm 0.0} & 72.12\textpm 0.1 & 57.02\textpm 0.9 & \textbf{0.0\textpm 0.0} & 69.89\textpm 0.5  & 62.93\textpm 4.2  &  \textbf{0.0\textpm 0.0}  & 73.10\textpm 2.3 \\\hline
ConRo &   \textbf{68.33\textpm 3.9} & 2.20\textpm 0.5  &  \textbf{90.50\textpm 0.3}  & \textbf{71.40\textpm 2.3} & 31.50\textpm 2.1 & \textbf{79.50\textpm 2.1} &   \textbf{77.56\textpm 2.3}  & 5.80\textpm 0.8  &  \textbf{97.10\textpm 0.4}  \\\hline
\end{tabular}
}
\end{table*}

The performance of our ConRo framework and baselines for all datasets are shown in Table \ref{tb:expr_results}. Clearly, our ConRo outperforms all baselines\footnote{Since DevNet classifies all test sessions as normal, we have not shown its performance scores. Devnet  does not employ any augmented malicious sessions for its training, so it cannot effectively address the dataset imbalance challenge.} w.r.t most of the performance metrics. These baselines do not learn effective class-specific shared features in the encoded representation space. Thus, due to the combined challenges of session diversity, dataset imbalance, and biased malicious training samples, they fail to provide noticeable results. However, ConRo addresses all these challenges effectively. It addresses the session diversity challenge through supervised contrastive learning. It addresses the dataset imbalance challenge by generating a large amount of augmented/potential malicious sessions. Finally, it addresses the challenge of biased malicious training samples by generating diverse potential malicious sessions. 

For the UMD-Wikipedia dataset, Swan noticeably outperforms our ConRo w.r.t FPR score. The mechanisms of ConRo and Swan are different. Specifically, Swan learns to identify unseen malicious sessions in a residual representation space $(\mathbf{v}_0-\mathbf{z})$ whereas, ConRo learns to identify unseen malicious sessions by generating a large amount of diverse potential malicious sessions. In UMD-Wikipedia dataset, many normal sessions share close similarities with unseen malicious sessions. Therefore, ConRo identifies some of the test normal sessions sharing close similarities with test malicious sessions as malicious (false positive) which negatively impacts the FPR scores of ConRo.  However, Swan does not specifically address the session diversity challenge in malicious sessions due to which, Swan under-performs against ConRo w.r.t F1 and AUCROC scores.

\subsubsection{Sensitivity Analysis}
\label{sec:sensitivity_analysis_study}

\begin{figure}
\begin{minipage}[t]{0.485\columnwidth}
  \includegraphics[width=0.75\linewidth]{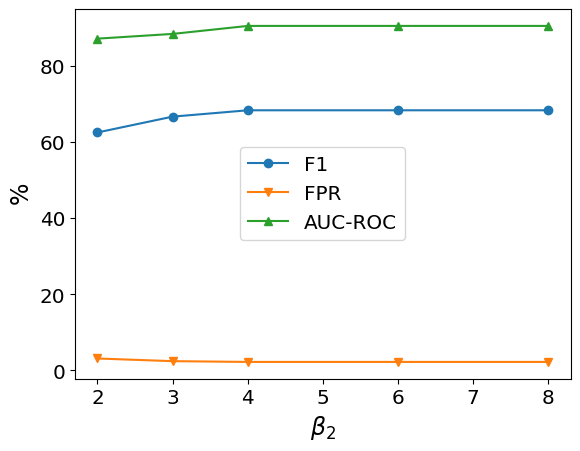}
\end{minipage}\hfill 
\begin{minipage}[t]{0.485\columnwidth}
  \includegraphics[width=0.75\linewidth]{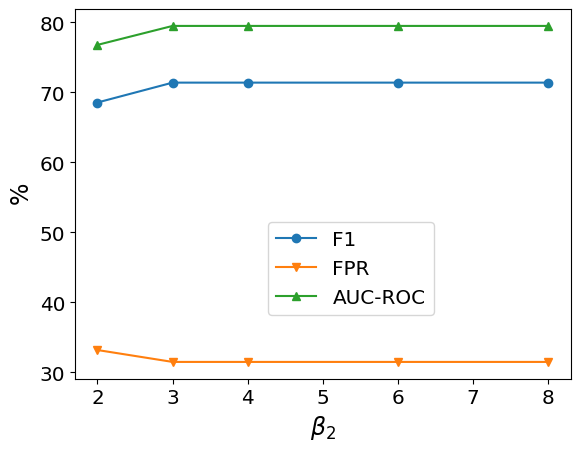}
\end{minipage}
\caption{Sensitivity analysis results w.r.t $\beta_2$ on CERT (left column) and UMD-Wikipedia (right column).}
\label{fig:sensitivty_analysis_result}
\end{figure}

For analyzing the sensitivity\footnote{We don't perform sensitivity analysis on $\beta_1$ because it is constrained to be set closer to 1~\citep{DBLP:conf/icml/VermaLKPL21}.} of the hyper-parameter $\beta_2$,  we first  perform first stage training of our encoder. Then by employing this first stage trained encoder, we further perform second stage training of our encoder separately corresponding to different $\beta_2$ values. Our sensitivity analysis results are shown in Figure \ref{fig:sensitivty_analysis_result}. Clearly, ConRo is not highly sensitive to hyper-parameter $\beta_2$. It only suffers slightly when $\beta_2$ is low and performance values converge for higher values. For example in the CERT dataset, for $\beta_2\geq 4$, the performance values converge. $\beta_2$ controls the diversity aspect of potential malicious sessions. The test set malicious sessions are not extremely diverse from their training set counterparts. Hence, generating extremely diverse potential malicious sessions does not aid in improving the generalization performance on the test set. In certain datasets, where the test case malicious sessions are extremely diverse when compared to their training set counterparts, we expect that higher values of $\beta_2$ can aid in improving the generalization performance.

\subsubsection{Visualization Analysis}

\begin{figure*}[htbp]
\centering
\begin{subfigure}{.30\textwidth}
  \centering
  \includegraphics[width=0.9\linewidth]{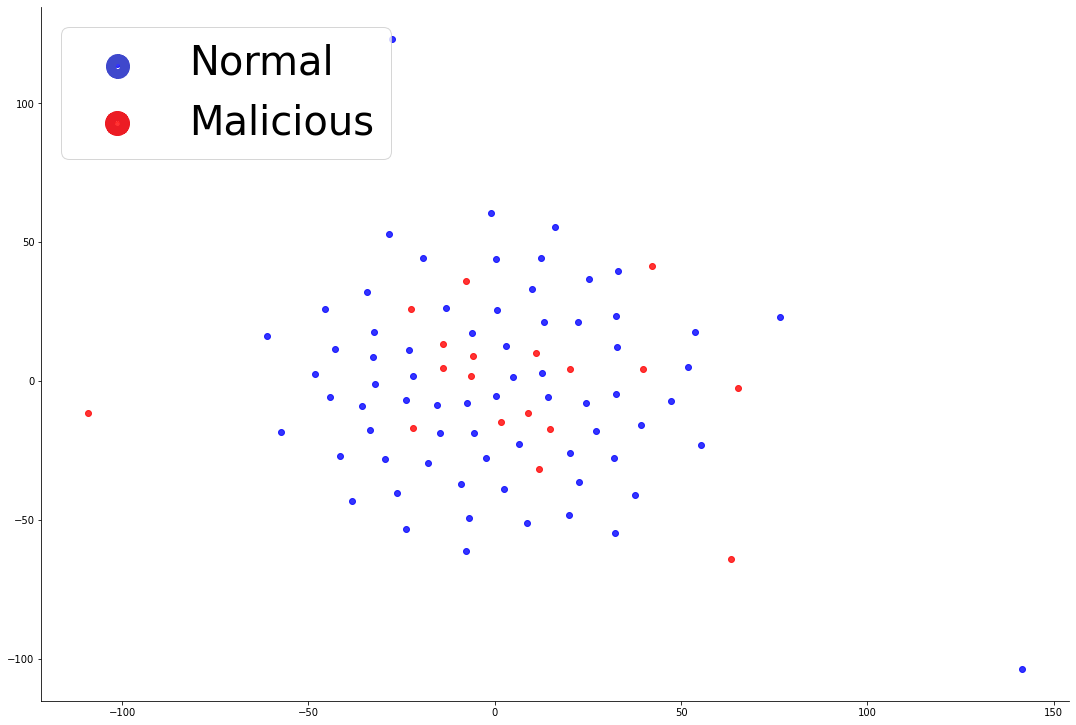}
  \caption{Raw sessions}
  \label{fig:raw_cert}
\end{subfigure}%
\label{fig:visualization_umd}
\begin{subfigure}{.30\textwidth}
  \centering
  \includegraphics[width=0.9\linewidth]{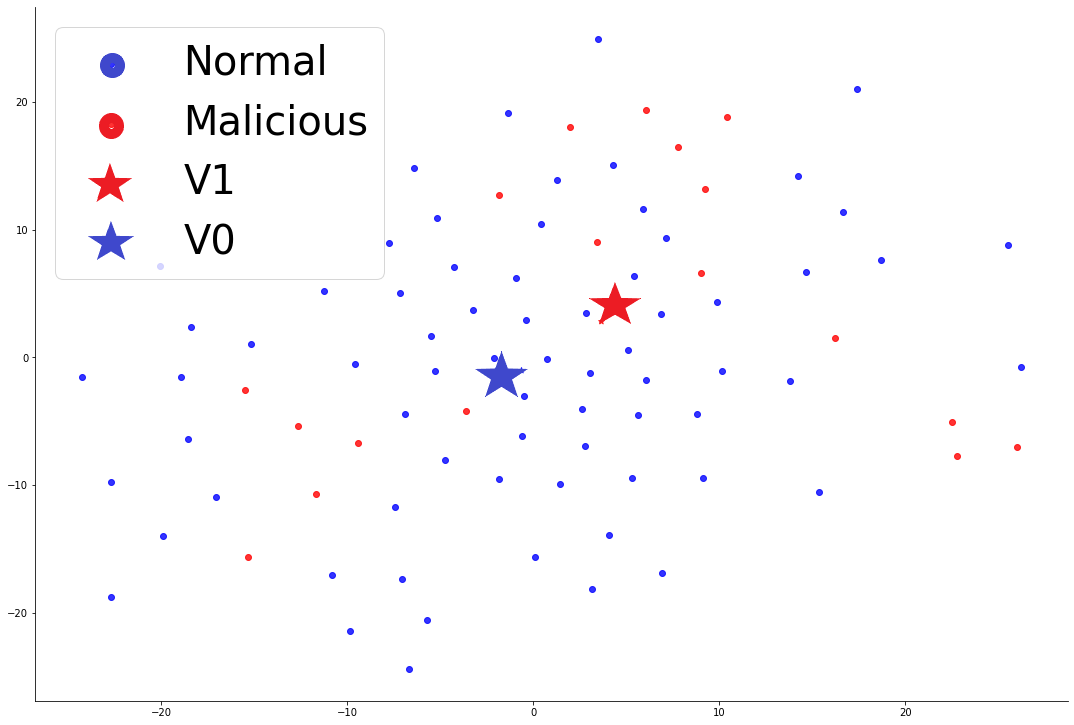}
  \caption{$\mathcal{L}_1$}
  \label{fig:stage1_cert}
\end{subfigure}
\begin{subfigure}{.30\textwidth}
  \centering
  \includegraphics[width=0.9\linewidth]{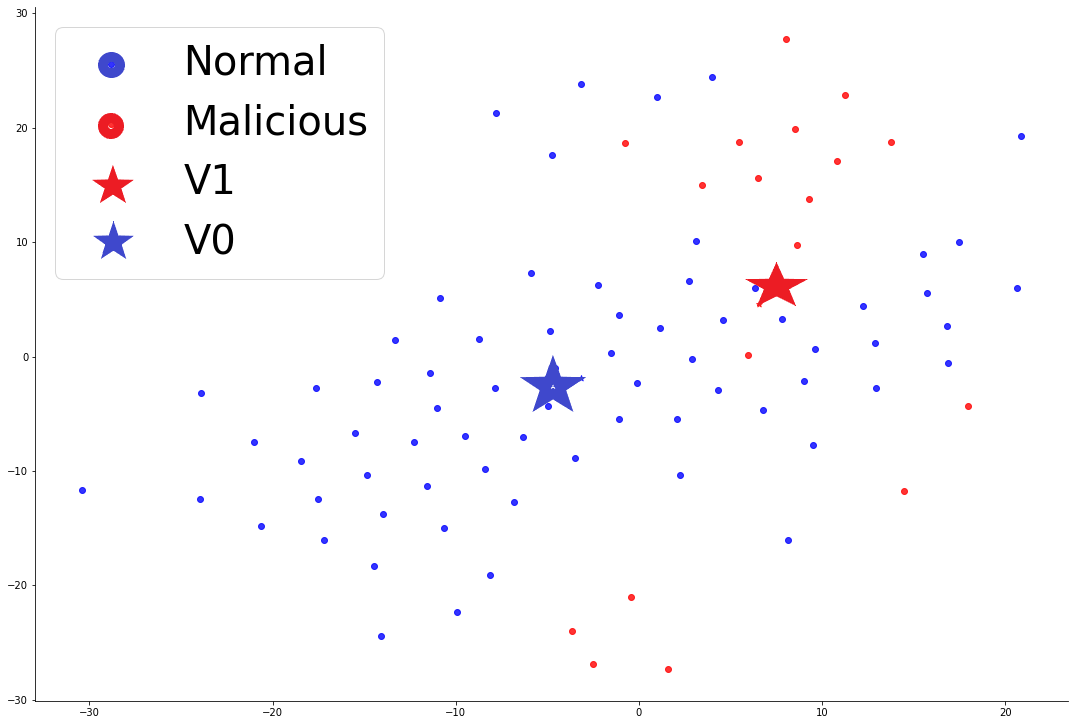}
  \caption{$\mathcal{L}_2$}
  \label{fig:stage2_cert}
\end{subfigure}
\caption{Visualization of test session representations for CERT dataset. Blue and red dots denote normal and malicious sessions, respectively. Blue and red stars denote $\mathbf{v}_0$ and $\mathbf{v}_1$, respectively.}
\label{fig:visualization_cert}
\end{figure*}

\begin{figure*}[htbp]
\centering
\begin{subfigure}{.32\textwidth}
  \centering
  \includegraphics[width=0.9\linewidth]{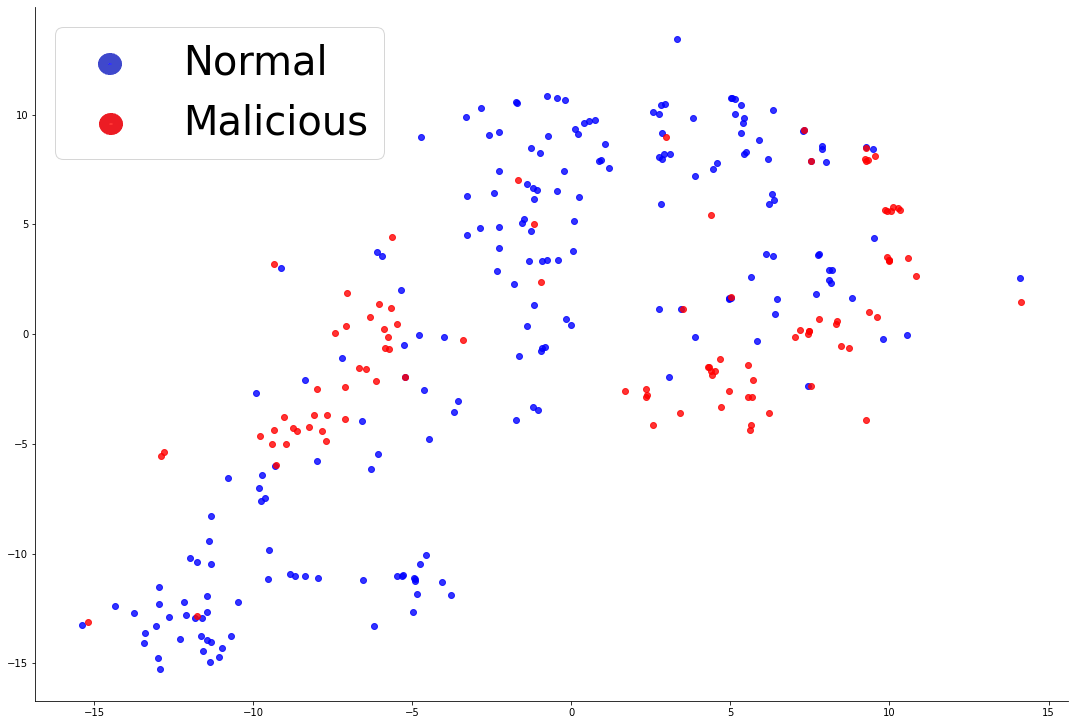}
  \caption{Raw sessions}
  \label{fig:raw_umd}
\end{subfigure}%
\label{fig:visualization_umd}
\begin{subfigure}{.32\textwidth}
  \centering
  \includegraphics[width=0.9\linewidth]{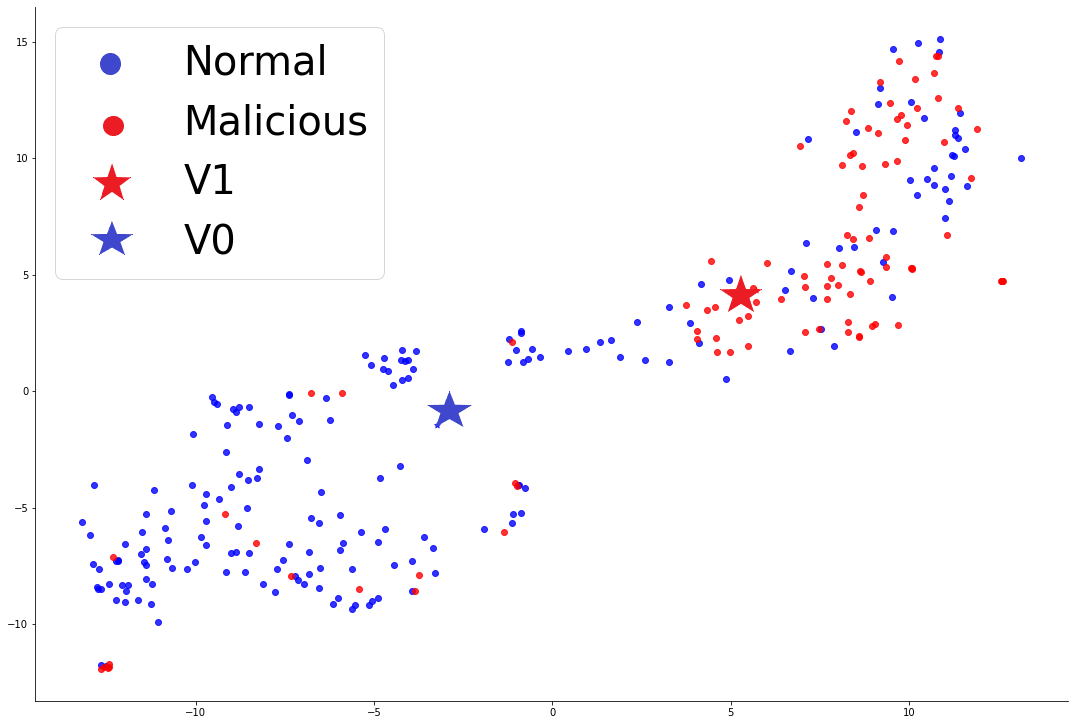}
  \caption{$\mathcal{L}_1$}
  \label{fig:stage1_umd}
\end{subfigure}
\begin{subfigure}{.32\textwidth}
  \centering
  \includegraphics[width=0.9\linewidth]{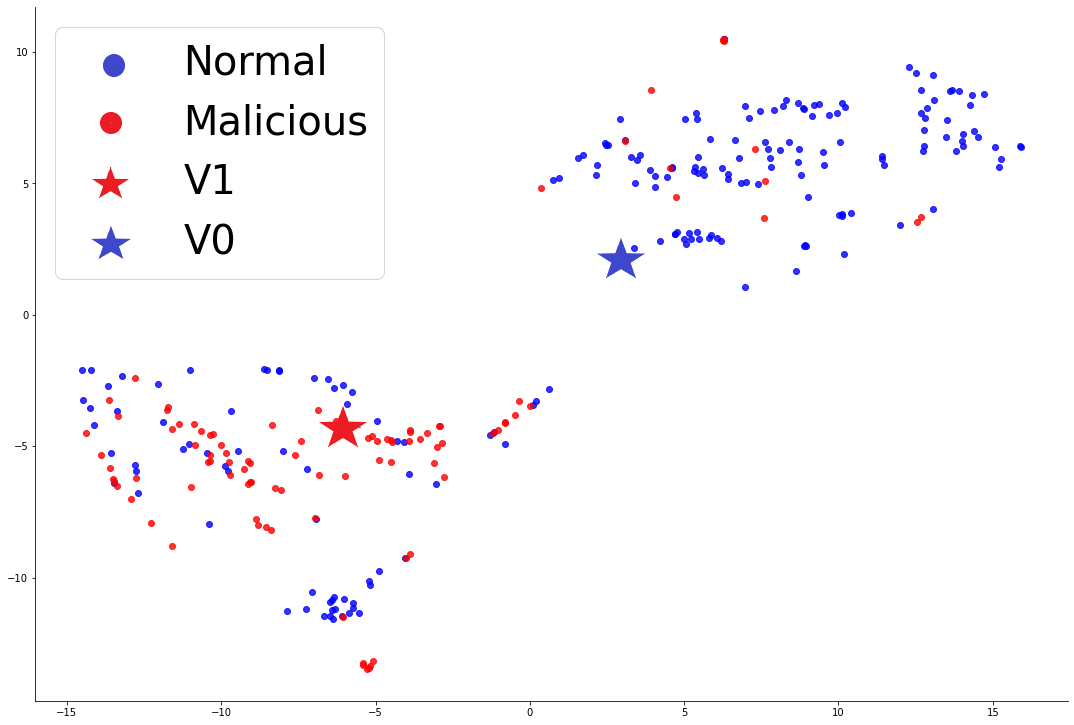}
  \caption{$\mathcal{L}_2$}
  \label{fig:stage2_beta2_4_umd}
\end{subfigure}
\caption{Visualization of test session representations for UMD-Wikipedia dataset.}
\label{fig:visualization_umd}
\end{figure*}

We employ the t-SNE technique~\citep{JMLR:v9:vandermaaten08a} for visualization. We consider two different encoded session representations based on whether our encoder is trained by $\mathcal{L}_1$ (stage 1) or $\mathcal{L}_2$ (stage 2). For the CERT dataset, we randomly sample 70 normal sessions from the test set and utilize all test malicious sessions for visualization. The visualization results for the CERT dataset are shown in Figure \ref{fig:visualization_cert}.  After stage 1 training (refer to Figure \ref{fig:stage1_cert}), there are many malicious sessions that overlap the normal session cluster. Also, $\mathbf{v}_1$ and $\mathbf{v}_0$ are closer to each other. The reason is that during stage 1 training, our encoder does not get an opportunity to contrast between unseen malicious and normal sessions. After stage 2 training (refer to Figure \ref{fig:stage2_cert}), $\mathbf{v}_1$ and $\mathbf{v}_0$ get well separated. During stage 2  training, we train our encoder by employing both similar and diverse potential malicious sessions. Therefore,  our encoder learns separable representations w.r.t the open-set fraud detection task.

For the UMD-Wikipedia dataset, we randomly sample 200 normal and 100 malicious sessions from the test set. The visualization results for UMD-Wikipedia dataset are shown in Figure \ref{fig:visualization_umd}.
After stage 1 training, as seen in Figure \ref{fig:stage1_umd}, even though $\mathbf{v}_1$ and $\mathbf{v}_0$ are not closer, the malicious and normal session clusters are not well separated. This visualization result demonstrates that stage 1 is not sufficient enough to obtain a good separation between normal and malicious sessions in the encoded representation space. After stage 2 training, as seen in Figure \ref{fig:stage2_beta2_4_umd}, both malicious and normal session clusters get well separated. 

\subsection{Training Latency Analysis}

\begin{table}[htp]
    \centering
    \caption{Training latencies of our ConRo and baselines.}
    \begin{tabular}{|c|c|c|}
    \hline
        \multirow{2}{*}{\textbf{Models}} & \multicolumn{2}{c|}{Training Latency (seconds)}  \\ \cline{2-3}
        & CERT & UMD-Wikipedia  \\
    \hline
        DeepSVDD & 2463  & 1322   \\
    \hline
        DeepSAD &  3842 &  2164  \\
    \hline
        DevNet & 1741  &  844  \\
    \hline
        CLDet &  6548  &  3864  \\    
    \hline
        Swan & 3920  &  2228  \\
    \hline
        ConRo &  40143  &  25080  \\
    \hline
    \end{tabular}%
    \label{tab:training_latency}
\end{table}

All experiments are executed on AMD EPYC (2.3 GHz) CPU server with 26 GB RAM and 226 GB hard disk. We use 150 epochs to train all baselines except CLDet. Specifically, CLDet has pre-training and fine tuning components which are trained by employing 10 and 150 epochs, respectively.  The training latencies of our ConRo framework and baselines for both CERT and UMD-Wikipedia datasets are shown in Table \ref{tab:training_latency}. ConRo incurs substantially more training cost than other baselines. The reason being that ConRo employs supervised contrastive loss which is the primary factor for this observed high training costs. However, supervised contrastive learning enables our session encoder to learn class-specific shared features,  and effectively address the session diversity challenge. Hence, ConRo is able to deliver better performance than other baselines.

\subsubsection{Ablation Analysis}
\label{sec:ablation_analysis}

\begin{table*}[ht]
\caption{Ablation analysis results (mean\textpm std).}
\label{tb:ablation_results}
\resizebox{1.0\textwidth}{!}{
\begin{tabular}{|c|c|c|c|c|c|c|c|c|c|}
\hline
\multirow{2}{*}{\textbf{Models}} & \multicolumn{3}{|c|}{\textbf{CERT}} & \multicolumn{3}{|c|}{\textbf{UMD-Wikipedia}} & \multicolumn{3}{|c|}{\textbf{Open-Stack}} \\ \cline{2-10}
 &   F1 & FPR & AUC-ROC & F1 & FPR & AUC-ROC &  F1 & FPR & AUC-ROC\\\hline
 w/o stage 1 &   18.33\textpm 1.2  &  25.10\textpm 0.3  & 78.56\textpm 1.1 &  40.23\textpm 1.3   &  28.37\textpm 1.9  &  55.55\textpm 1.1 &  14.92\textpm 1.7  &  47.14\textpm 2.3  & 49.43\textpm 2.9 \\\hline
  w/o $\mathcal{L}^{Sup}$ &    5.28\textpm 0.2  &  48.10\textpm 1.8  & 45.44\textpm 0.9  &  53.16\textpm 2.3  & 25.10\textpm 1.9  & 64.65\textpm 1.9 &  15.12\textpm 1.4  &  46.90\textpm 2.7  &  49.78\textpm 2.4\\\hline
 w/o $\mathcal{L}^{SV}$ &   20.10\textpm 1.1  &  27.05\textpm1.3  &  83.72\textpm 0.7  &   64.95\textpm 0.8  &  18.85\textpm 6.1  &  73.72\textpm 0.5 &  38.76\textpm 0.8   &  31.60\textpm 1.1  &  84.20\textpm 0.4  \\\hline
 w/o AO &   8.99\textpm 0.4  &  68.46\textpm 3.1  &  62.98\textpm 1.6  &  52.04\textpm 1.4  &  80.30\textpm 1.4   &  55.68\textpm 1.9   &  14.08\textpm 2.1 &  26.16\textpm 1.4 & 50.59\textpm 2.3  \\\hline
w/o stage 2 &    42.86\textpm 1.1   &  0.0\textpm 0.0  &  63.84\textpm 0.1   &  60.90\textpm 1.1  & 23.85\textpm 1.6  & 70.42\textpm 0.6 &  46.11\textpm 3.4  &  0.0\textpm 0.0 & 65.40\textpm 1.6 \\\hline
w/o $fp(\cdot)$ &   31.32\textpm 5.1   &  26.66\textpm 6.3  &  72.77\textpm 3.1  &  59.38\textpm 1.8  & 65.52\textpm 4.2  & 65.95\textpm 2.6  & 37.50\textpm 3.7  &  49.60\textpm 3.6  & 48.16\textpm 3.5 \\\hline
w/o $\mathbf{\widehat{G}^1(\cdot)}$ &   55.92\textpm 1.2  &  4.13\textpm 0.3  &  89.49\textpm 0.2   & 65.58\textpm 2.6  & 31.20\textpm 0.3  & 74.04\textpm 2.3   &  67.57\textpm 1.1  &  9.60\textpm 0.4 &  95.20\textpm 0.2 \\\hline
w/o $\mathbf{\widetilde{G}^1(\cdot)}$ &   44.17\textpm 1.9   &  7.10\textpm 0.6  &  88.16\textpm 0.3  &  63.40\textpm 0.8 &  31.70\textpm 4.7 &  72.10\textpm 0.6 &  52.26\textpm 1.6 & 18.10\textpm 1.4  & 90.61\textpm 0.2 \\\hline
\end{tabular}
}
\end{table*}

We conduct the ablation analysis study on our ConRo framework by ablating the following main components: stage 1, $\mathcal{L}^{Sup}$, $\mathcal{L}^{SV}$, Alternating Optimization (AO), stage 2, $fp(\cdot)$ (optimistic choice), $\widehat{G}^1(\cdot)$, and $\widetilde{G}^1(\cdot)$. The ablation analysis results are shown in Table \ref{tb:ablation_results}. 

\noindent{\textbf{W/o stage 1}.} Mean F1 scores drop to 18.33 (CERT), 40.23 (UMD-Wikipedia), and 14.92 (Open-Stack).  Stage 1 ensures that the encoder learns shared features for normal sessions. Without learning these shared features, the encoder fails to achieve tight class-specific clusters in the encoded representation space.

\noindent{\textbf{W/o $\mathcal{L}^{Sup}$}.}  Mean F1 scores drop to 5.28 (CERT), 53.16 (UMD-Wikipedia), and 15.12 (Open-Stack). Both normal and malicious sessions typically exhibit large diversity and $\mathcal{L}^{Sup}$ is essential to address this session diversity challenge. We can see that there is a significant drop in F1 scores on CERT and Open-Stack datasets but not in the case for UMD-Wikiepdia dataset. We can attribute the reason to the different characteristics of these datasets. Addressing the session diversity challenge for the normal sessions is much more critical in both CERT  and Open-Stack datasets than in the UMD-Wikipedia dataset.

\noindent{\textbf{W/o $\mathcal{L}^{SV}$}.}  Mean F1 scores drop to 20.10 (CERT), 64.95 (UMD-Wikipedia), and 38.76 (Open-Stack). The DeepSVDD loss $\mathcal{L}^{SV}$ enables the encoder to push normal sessions in a minimum volume hyper-sphere in the encoded representation space. Without this topological effect, the efficacy of stage 2 reduces because the generated diverse potential malicious sessions do not effectively cover unseen malicious sessions. 

\noindent{\textbf{W/o AO }.} By employing the joint optimization approach, mean F1 scores drop to 8.99 (CERT), 52.04 (UMD-Wikipedia), and 14.08 (Open-Stack). Optimizing DeepSVDD objective $(\mathcal{L}^{SV})$ can yield maximum benefits only when the input normal sessions have considerable shared features in the encoded representation space. Here, we jointly optimize both supervised contrastive $(\mathcal{L}^{Sup})$ and DeepSVDD objectives, and we do not specifically provide normal sessions having considerable shared features in the encoded representation space as inputs to the DeepSVDD objective. As a consequence, we can observe a significant drop in the performance.

\noindent{\textbf{W/o stage 2 }.}  Mean F1 scores drop to 42.86 (CERT), 60.90 (UMD-Wikipedia), and 46.11 (Open-Stack). In stage 1 training, our encoder learns to contrast normal sessions with few available malicious sessions having limited diversity. Stage 2 generates diverse potential malicious sessions which can be similar to  unseen malicious sessions w.r.t their encoded representations. As a consequence, our encoder can learn effective separable encoded representations.

\noindent{\textbf{W/o} {$\mathbf{fp(\cdot)}$.}  By employing the pessimistic choice, mean F1 scores drop to 31.32 (CERT), 59.38 (UMD-Wikipedia), and 37.50 (Open-Stack). Without employing $fp(\cdot)$, the encoder learns to push malicious sessions and those potential malicious sessions which are false positives, closer in the encoded representation space. Due to this improper learning effect, the encoder does not achieve effective separable encoded representations.   

\noindent{\textbf{W/o $\mathbf{\widehat{G}^1(\cdot)}$}.}  Mean F1 scores drop to 55.92 (CERT), 65.58 (UMD-Wikipedia), and 67.57 (Open-Stack). Generating similar potential malicious sessions which are similar to a seen malicious session in the encoded representation space, aids the encoder to learn more effective separable representations. 
 
\noindent{\textbf{W/o $\mathbf{\widetilde{G}^1(\cdot)}$}.}  Mean F1 scores drop to 44.17 (CERT), 63.40 (UMD-Wikipedia), and 52.26 (Open-Stack). Generating diverse potential malicious sessions which can be similar to unseen malicious sessions in the encoded representation space, aids the encoder to effectively contrast normal sessions with unseen malicious sessions.

\section{Conclusion}

In this work, we have developed a robust and open-set fraud detection framework called ConRo, which is specifically designed to operate in the scenario where only a few malicious sessions having limited diversity are available for training. We developed a training procedure for ConRo to learn separable session representations by employing effective data augmentation strategies and by the combined effect of supervised contrastive and DeepSVDD losses. We presented a theoretical analysis study to analyze the main factors influencing the generalization performance of ConRo. The empirical study on three benchmark datasets demonstrated that our ConRo can outperform state-of-the-art baselines. In our future work, we plan to extend ConRo to address specific distribution shift scenarios such as \textit{sample selection bias}. We will study how to integrate  bias correction approaches with supervised contrastive learning.

\section*{Acknowledgement}
This work was supported in part by NSF grants  1920920, 1946391 and 2103829.

\bibliographystyle{unsrtnat}
\bibliography{Remote}

\end{document}